\documentclass[a4paper,11pt,twoside,onecolumn]{article}%
\usepackage{amsfonts, amsmath, amsthm, amssymb, graphicx}

\newcommand{\comment}[1]{}

\newtheorem{theorem}{Theorem}

\newtheorem{fact}[theorem]{Fact}
\newtheorem{claim}[theorem]{Claim}
\newcommand{\claimproof}[2]%
{\noindent{\em Proof of Claim \ref{#1}.}
#2\hspace*{\fill}$\Box$~~~~~\vspace{5mm} }

\newtheorem{definition}[theorem]{Definition}

\newtheorem{lemma}[theorem]{Lemma}

\newtheorem{rem}[theorem]{Remark}

\renewcommand{\mod}{\text{mod~}}
\newcommand{\poly}{\text{poly}}

\renewcommand{\exp}{\text{exp}}

\newcommand{\sps}{\Sigma\Pi\Sigma}
\newcommand{\FF}{\mathbb{F}}

\newcommand{\calT}{\mathcal{T}}

\def\vec#1{\overline{#1}}

\def\({\left(}
\def\){\right)}
\def\<{\langle}
\def\>{\rangle}

\def\ge{\geqslant}

\begin{document}

\title{\sc An Almost Optimal Rank Bound for Depth-3 Identities} 

\author{
Nitin Saxena
\footnote{Hausdorff Center for Mathematics, Bonn 53115, Germany. 
E-mail: {\tt ns@hcm.uni-bonn.de}
}
\and
C. Seshadhri
\footnote{IBM Almaden Research Center, San Jose - 95123, USA.
E-mail: {\tt csesha@us.ibm.com}
}
}
\date{}
\maketitle
\begin{abstract}
We show that the rank of a depth-$3$ circuit (over any field) that is simple, 
minimal and zero is at most $O(k^3\log d)$. The previous best rank bound known was 
$2^{O(k^2)}(\log d)^{k-2}$ by Dvir and Shpilka (STOC 2005). 
This almost resolves the rank question first posed by Dvir and Shpilka
(as we also provide a simple and minimal identity of rank $\Omega(k\log d)$). 

Our rank bound significantly improves (dependence on $k$ exponentially reduced) 
the best known deterministic black-box identity tests for depth-$3$ circuits 
by Karnin and Shpilka (CCC 2008). Our techniques also shed light 
on the factorization pattern of nonzero depth-$3$ circuits, most
strikingly: the rank of linear factors of a simple, minimal and nonzero depth-$3$ circuit 
(over any field) is at most $O(k^3\log d)$.

The novel feature of this work is a new notion of maps between sets
of linear forms, called \emph{ideal 
matchings}, used to study depth-$3$ circuits. We prove interesting structural results about
depth-$3$ identities using these techniques. We believe that these can lead
to the goal of a deterministic polynomial time identity test for these circuits.

\end{abstract}

\section{Introduction}

Polynomial identity testing (PIT) ranks as one of the most important open problems 
in the intersection of algebra and computer science. 
We are provided an arithmetic circuit that computes a polynomial $p(x_1,x_2,\cdots,x_n)$
over a field $\mathbb{F}$, and we wish to test if $p$ is identically zero. 
In the black-box setting, the circuit
is provided as a black-box and we are only allowed to evaluate the polynomial
$p$ at various domain points. 
The main goal is to devise a deterministic polynomial time
algorithm for PIT.
Kabanets and Impagliazzo~\cite{KI04} and Agrawal~\cite{Ag05}
have shown connections between deterministic algorithms for identity testing and circuit lower bounds,
emphasizing the importance of this problem.

The first randomized polynomial time PIT algorithm, 
which was a black-box algorithm,
was given (independently) by Schwartz~\cite{sch80} and Zippel~\cite{Z79}.
Randomized algorithms that use less
randomness were given by Chen \& Kao~\cite{CK00}, Lewin \& Vadhan~\cite{LV98},
and Agrawal \& Biswas~\cite{AB03}.
Klivans and Spielman~\cite{ks01} observed that 
even for depth-$3$ circuits for bounded top fanin,
deterministic identity testing was open. Progress
towards this was first made by Dvir and Shpilka~\cite{DS06},
who gave a quasi-polynomial time algorithm, although with
a doubly-exponential dependence on the top fanin.
The problem was resolved by a polynomial time algorithm given by
Kayal and Saxena~\cite{ks07}, with a running time exponential
in the top fanin. For a special case of 
depth-$4$ circuits, Saxena~\cite{s08} has designed
a deterministic polynomial time algorithm for PIT.
Why is progress restricted to small depth circuits?
Agrawal and Vinay~\cite{AV08} recently showed that
an efficient black-box identity test for depth-$4$
circuits will actually give a quasi-polynomial black-box test
for circuits of \emph{all depths}.

For deterministic black-box testing, the first results
were given by Karnin and Shpilka~\cite{KSh08}. 
Based on results in~\cite{DS06},
they gave an algorithm for depth-$3$ circuits having a quasi-polynomial
running time (with a doubly-exponential dependence
on the top fanin)\footnote{\cite{KSh08} had a better running time for read-$k$ 
depth-$3$ circuits, where each variable appears at most $k$ times. But even there 
the dependence on $k$ is doubly-exponential.}.
One of the consequences of our result will be a significant improvement
in the running time of their deterministic black-box tester.

This work focuses on depth-$3$ circuits. A structural study of 
depth-$3$ identities was initiated in \cite{DS06} by defining a notion of {\em rank} of {\em simple} and {\em 
minimal} identities. A depth-$3$ circuit $C$ over a field $\FF$ is: 
$$C(x_1,\ldots,x_n) = \sum_{i=1}^k T_i$$ 
where, $T_i$ ({\em a multiplication term}) is a product of $d_i$ linear functions 
$\ell_{i,j}$ over $\FF$. Note that for the purposes of studying identities we can 
assume wlog (by {\em homogenization}) that $\ell_{i,j}$'s are linear {\em forms} 
(i.e. linear polynomials with a zero constant coefficient) and that $d_1=\cdots=d_k=:d$. 
Such a circuit is referred to as a $\sps(k,d)$ circuit,
where $k$ is the \emph{top fanin} of $C$ and $d$ is the {\em degree} of $C$. 
We give a few definitions from~\cite{DS06}.

\begin{definition} {\bf[Simple Circuits]}
$C$ is a {\em simple} circuit if there is no nonzero linear form dividing all 
the $T_i$'s. 

{\bf[Minimal Circuits]} $C$ is a {\em minimal} circuit if for every proper subset 
$S\subset[k]$, $\sum_{i\in S}T_i$ is nonzero. 

{\bf[Rank of a circuit]} The {\em rank} of the circuit, $rank(C)$, is defined as 
the rank of the linear forms $\ell_{i,j}$'s viewed as $n$-dimensional vectors over 
$\FF$.
\end{definition}

Can all the forms $\ell_{i,j}$ be independent, or must there
be relations between them?
The rank can be interpreted as the minimum number of variables
that are required to express $C$. There exists a linear
transformation converting the $n$ variables of the circuit
into $rank(C)$ independent variables. A trivial bound on the rank (for any $\sps$-circuit) 
is $kd$, since that is the total number of
linear forms involved in $C$. 
The rank is a fundamental property of a $\sps(k,d)$ circuit
and it is crucial to understand how large this can be 
for identities.
A substantially smaller rank bound than $kd$ shows that 
identities do not have as many ``degrees of freedom" as general
circuits, and lead 
to deterministic identity tests\footnote{We
usually do not get a polynomial time algorithm.}. 
Furthermore, the techniques used to prove rank bounds 
show us structural properties of identities that may 
suggest directions to resolve PIT for $\sps(k,d)$ circuits.

Dvir and Shplika~\cite{DS06} proved that the
rank is bounded by $2^{O(k^2)}(\log d)^{k-2}$, and this bound
is translated to a $\poly(n) \exp(2^{O(k^2)}(\log d)^{k-1})$
time black-box identity tester by Karnin and Shpilka~\cite{KSh08}.
Note that when $k$ is larger than $\sqrt{\log d}$, these
bounds are trivial. 

Our present understanding of $\sps(k,d)$ identities is very poor when
$k$ is larger than a constant. We present the
first result in this direction.

\begin{theorem}[Main Theorem] \label{thm-main} The rank of a simple and minimal 
$\sps(k,d)$ identity is $O(k^3 \log d)$.
\end{theorem}

This gives an exponential improvement on the previously known dependence on $k$, 
and is strictly better than the previous rank bound for every $k > 3$.
We also give a simple construction of identities with rank
$\Omega(k\log d)$ in Section~\ref{sec-high-rank-eg}, showing
that the above theorem is almost optimal. As mentioned above, we can interpret
this bound as saying that any simple and minimal $\sps(k,d)$
identity can be expressed using $O(k^3\log d)$ independent variables.
One of the most interesting features of this result is a novel
technique developed to study depth-$3$ circuits. We introduce the
concepts of \emph{ideal matchings} and \emph{ordered matchings},
that allow us to analyze the structure of
depth-$3$ identities. These matchings are studied
in detail to get the rank bound. Along the way we initiate a theory of 
matchings, viewing a matching as a fundamental map between sets of linear forms.

Why are the simplicity and minimality restrictions required? 
Take the non-simple $\sps(2,d)$ identity $(x_1x_2\cdots x_d) - (x_1x_2\cdots x_d)$.
This has rank $d$.
Similarly,
we can take the non-minimal $\sps(4,d+1)$ identity $(y_1 y_2 \cdots y_d)(x_1 - x_1) 
+ (z_1 z_2 \cdots z_d)(x_2 - x_2)$ that has rank $(2d+2)$. 
In some sense, these restrictions only ignore identities that
are composed of smaller identities.

\subsection{Consequences}

Apart from being an interesting structural result about $\sps$ identities,
we can use the rank bound to get nice algorithmic results.
Our rank bound immediately gives faster deterministic black-box identity
testers for $\sps(k,d)$ circuits. 
A direct application of Lemma 4.10 in \cite{KSh08}
to our rank bound
gives an exponential improvement in the dependence of $k$
compared to previous black-box testers
(that had a running time of $\poly(n) \exp(2^{O(k^2)}(\log d)^{k-1})$).

\begin{theorem} \label{thm-blackbox} 
There is a deterministic black-box identity
tester for $\sps(k,d)$ circuits that runs in $\poly(n,d^{k^3\log d})$ time.
\end{theorem}

The above black-box tester is now much closer in complexity to the best {\em non}
black-box tester known ($poly(n,d^k)$ time by \cite{ks07}). 

Our result also applies to black-box identity testing of 
\emph{read-$k$} $\sps(k,d)$ circuits,
where each variable occurs at most $k$ times. We get a similar immediate
improvement in the dependence of $k$ (the previous running time
was $n^{2^{O(k^2)}}$.)

\begin{theorem} \label{thm-readk} There is a deterministic black-box identity
tester for read-$k$ $\sps(k,d)$ circuits that runs in
$O(n^{k^4\log k})$ time.
\end{theorem}

Although it is not immediate from Theorem~\ref{thm-main},
our technique also provides an interesting algebraic result about
polynomials computed by simple, minimal, and nonzero $\sps(k,d)$ 
circuits\footnote{Here we can also consider circuits where
the different terms in $C$ have different degrees. The parameter $d$ 
is then an upper bound on the degree of $C$.}.
Consider such a circuit $C$ that computes a polynomial $p(x_1,\cdots,x_n)$.
Let us factorize $p$ into $\prod_i q_i$, where each $q_i$ is a nonconstant and 
irreducible polynomial.
We denote by $L(p)$ the set of \emph{linear factors}
of $p$ (that is, $q_i \in L(p)$ iff $q_i|p$ is linear).

\begin{theorem}\label{thm-factors}  If $p$ is computed by a simple, minimal, nonzero 
$\sps(k,d)$ circuit then the rank of $L(p)$ is at most $k^3\log d$.
\end{theorem}

\subsection{Organization}

We first give a simple construction of identities with rank
$\Omega(k\log d)$ in Section~\ref{sec-high-rank-eg}. 
Section~\ref{sec-rank} contains the proof 
of our main theorem. We give some preliminary notation
in Section~\ref{subsec-prelim}
before explaining an intuitive picture of our ideas (Section~\ref{subsec-intuition}). We
then explain our main tool of \emph{ideal matchings} (Section~\ref{subsec-ideal-mat})
and prove some useful lemmas about them.
We move to Section~\ref{subsec-ordered-mat} where the concepts of \emph{ordered matchings}
and \emph{simple parts of circuits} are introduced. We motivate
these definitions and then prove some easy facts about
them. We are now ready to tackle the problem
of bounding the rank. We describe our proof in terms
of an iterative procedure in Section~\ref{subsec-useful}.
Everything is put together in Section~\ref{subsec-count}
to bound the rank.
Finally (it should hopefully
be obvious by then), we show how to apply our 
techniques to prove Theorem~\ref{thm-factors}.

\section{High Rank Identities}\label{sec-high-rank-eg}

The following identity was constructed in \cite{ks07}: over $\FF_2$ (with $r\ge2$),
\begin{align*}
C(x_1,\ldots,x_r) :=\quad 
& \prod_{\substack{b_1,\ldots,b_{r-1}\in\FF_2\\b_1+\cdots+b_{r-1}\equiv1}}
(b_1x_1+\cdots+b_{r-1}x_{r-1})\\
&+\ \prod_{\substack{b_1,\ldots,b_{r-1}\in\FF_2\\b_1+\cdots+b_{r-1}\equiv0}}
(x_r+b_1x_1+\cdots+b_{r-1}x_{r-1})\\ 
&+\ \prod_{\substack{b_1,\ldots,b_{r-1}\in\FF_2\\b_1+\cdots+b_{r-1}\equiv1}}
(x_r+b_1x_1+\cdots+b_{r-1}x_{r-1})
\end{align*}
It was shown that, over $\FF_2$, $C$ is a simple and minimal $\sps$ zero circuit 
of degree $d=2^{r-2}$ with $k=3$ multiplication terms and $rank(C)=r=\log_2 d+2$. 
For this section let $S_1(\vec{x})$, $S_2(\vec{x})$, $S_3(\vec{x})$ denote the 
three multiplication terms of $C$. We now build a high rank identity based on 
$S_1, S_2, S_3$. Our basic step is given by the following lemma that was used in 
\cite{DS06} to construct identities of rank $(3k-2)$.

\begin{lemma}\cite{DS06}\label{lem-join-of-ids}
Let $D_i(y_{i,1},\ldots,y_{i,r_i}):=$ $\sum_{j=1}^{k_i}T_j$ be a simple, minimal 
and zero $\sps$ circuit, over $\FF_2$, with degree $d_i$, fanin $k_i$ and rank 
$r_i$. Define a new circuit over $\FF_2$ using $D_i$ and $C$:
\begin{align*}
D_{i+1}(y_{i,1},\ldots,y_{i,r_i+r}):=\quad 
& \(\sum_{j=1}^{k_i-1}T_j\)\cdot S_1(y_{i,r_i+1},\ldots,y_{i,r_i+r}) - 
T_{k_i}\cdot S_2(y_{i,r_i+1},\ldots,y_{i,r_i+r}) \\
&\quad - T_{k_i}\cdot S_3(y_{i,r_i+1},\ldots,y_{i,r_i+r})
\end{align*}
Then $D_{i+1}$ is a simple, minimal and zero $\sps$ circuit with degree 
$d_{i+1}=(d_i+d)$, fanin $k_{i+1}=(k_i+1)$ and rank $r_{i+1}=(r_i+r)$.
\end{lemma} 

\begin{proof} Since $C$ is an identity, we get that 
$S_2(y_{i,r_i+1},\ldots,y_{i,r_i+r}) + S_3(y_{i,r_i+1},\ldots,y_{i,r_i+r}) = - S_1(y_{i,r_i+1},\ldots,y_{i,r_i+r})$.
Therefore, 
\begin{eqnarray*}
& & D_{i+1}(y_{i,1},\ldots,y_{i,r_i+r}) \\
& = & \big(\sum_{j=1}^{k_i-1}T_j\big) S_1(y_{i,r_i+1},\ldots,y_{i,r_i+r}) - 
T_{k_i}\left( S_2(y_{i,r_i+1},\ldots,y_{i,r_i+r}) + S_3(y_{i,r_i+1},\ldots,y_{i,r_i+r}) \right)\\
& = & \big(\sum_{j=1}^{k_i-1}T_j\big)\cdot S_1(y_{i,r_i+1},\ldots,y_{i,r_i+r}) + 
T_{k_i} S_1(y_{i,r_i+1},\ldots,y_{i,r_i+r}) \\
& = & \big(\sum_{j=1}^{k_i}T_j\big) \cdot S_1(y_{i,r_i+1},\ldots,y_{i,r_i+r}) = 0
\end{eqnarray*}

The terms $T_j$ do not share any variables with $S_\ell$ ($\ell \in \{1,2,3\}$).
Since $D_i$ and $C$ are simple, $D_{i+1}$ is also simple. Suppose
$D_{i+1}$ is not minimal. We have some subset $P \subset [1,k_i-1]$
such that $C' := (\sum_{j \in P} T_j) S_1 - \alpha_2 T_{k_i}S_2 - \alpha_3 T_{k_i}S_3 = 0$,
where $\alpha_2, \alpha_3 \in \{0,1\}$. If both $\alpha_2$ and $\alpha_3$
are $1$, then we get $(\sum_{j \in P} T_j) S_1 + T_{k_i}S_1 = 0$, now $P$ must be the whole 
set $[1,k_i-1]$, because $D_i$ is minimal. On the other hand,
if both $\alpha_2,\alpha_3$ are $0$, then $(\sum_{j \in P} T_j) S_1 = 0$ which is impossible
as $D_i$ is minimal. The only remaining possibility is (wlog)
$(\sum_{j \in P} T_j) S_1 - T_{k_i}S_2 = 0$. As $S_1$ is coprime to $S_2$ and $T_{k_i}$, 
this is impossible. Therefore, $D_{i+1}$ is minimal.

It is easy to see the parameters of $D_{i+1}$: $k_{i+1}=(k_i + 1)$ and $d_{i+1}=(d_i+1)$.
Because the $T_j$'s do not share any variables with $S_\ell$'s,
the rank $r_{i+1} = (r_i + r)$.
\end{proof}

{\bf Family of High Rank Identities: } 
Now we will start with $D_0:=C(y_{0,1},\ldots,y_{0,r})$ and apply the above 
lemma iteratively. The $i$-th circuit we get is $D_i$ with degree $d_i=(i+1)d$, 
fanin $k_i=i+3$ and rank $r_i=(i+1)r$ $=(i+1)(\log_2 d+2)$. So $r_i$ relates 
to $k_i, d_i$ as:
$$r_i = (k_i-2)\(\log_2\frac{d_i}{k_i-2}+2\).$$
Also it can be seen that if $d>i$ then $\frac{d_i}{k_i-2}\geq\sqrt{d_i}$. Thus
after simplification, we have for any $3\leq i<d$, $r_i>$ 
$\frac{k_i}{3}\cdot\log_2 d_i$. This gives us an infinite family of $\sps(k,d)$ 
identities over $\FF_2$ with rank $\Omega(k\log d)$. A similar family can be 
obtained over $\FF_3$ as well.

\section{Rank Bound}\label{sec-rank}

Our technique to bound the rank of $\sps$ identities relies mainly on two 
notions - {\em form-ideals} and {\em matchings} by them - that occur naturally 
in studying a $\sps$ circuit $C$. Using these tools we can do a surgery on the
circuit $C$ and extract out smaller circuits and smaller identities. Before 
explaining our basic idea we need to develop a small theory of matchings and 
define {\em gcd and simple} parts of a {\em subcircuit} in that framework.

We set down some preliminary definitions before
giving an imprecise, yet intuitive explanation of our idea and
an overall picture of how we bound the rank.

\subsection{Preliminaries}\label{subsec-prelim}

We will denote the set $\{1,\ldots,n\}$ by $[n]$. 

In this paper we will study identities over a field $\FF$. So the circuits
compute multivariate polynomials in the {\em polynomial ring} $R:=$ 
$\FF[x_1,\ldots,x_n]$. We will be studying $\sps(k,d)$ {\em circuits} :
such a circuit $C$ is an expression in $R$ given by a depth-$3$
circuit, with the top gate being an addition gate, the
second level having multiplication gates, the last level having
addition gates, and the leaves being variables. The edges
of the circuit have elements of $\FF$ (constants) associated with them
(signifying multiplication by a constant). The top fanin
is $k$ and $d$ is the degree of the polynomial computed by $C$.
We will call $C$ a $\sps$-\emph{identity}, if $C$ is an identically
zero $\sps$-circuit.

A \emph{linear form} is a linear polynomial in $R$. We will denote
the set of all linear forms by $L(R)$ : 
$$L(R):= \left\{\sum_{i=1}^n a_ix_i \mid a_1,\ldots,a_n\in\FF\right\}$$
Much of what we do shall deal with sets of linear forms, and
various maps between them. A \emph{list} $L$ of linear
forms is a multi-set of forms with an arbitrary order
associated with them. The actual ordering is unimportant : we
merely have it to distinguish between repeated forms
in the list. One of the fundamental constructs we use
are maps between lists, which could have many
copies of the same form. The ordering allows us to
define these maps unambiguously. All lists we consider
will be finite.

\begin{definition} {\bf[Multiplication term]} A {\em multiplication term} $f$ is an 
expression in $R$ given as (the product may have repeated $\ell$'s):
$$f := c\cdot\prod_{\ell\in S}\ell,\text{~~~ where }c\in\FF^*
\text{ and } S \text{ is a list of linear forms.} 
$$
The {\em list of linear forms in $f$}, $L(f)$, is just the list $S$ of forms 
occurring in the product above. $\#L(f)$ is naturally called the {\em degree} 
of the multiplication term. 
For a list $S$ of linear forms we define the {\em multiplication 
term of $S$}, $M(S)$, as $\prod_{\ell\in S}\ell$ or $1$ if $S=\phi$. 
\end{definition}

\begin{definition}{\bf[Forms in a Circuit]}
We will represent a $\sps(k,d)$ circuit $C$ as a sum of $k$ multiplication terms
of degree $d$, $C = \sum_{i=1}^k T_i$. The list of {\em linear forms occurring in $C$} 
is $L(C):=$ $\bigcup_{i\in[k]}L(T_i)$. Note that $L(C)$ is a list of size exactly $kd$. 
The {\em rank of $C$}, $rank(C)$, is just the number of linearly independent
linear forms in $L(C)$.
\end{definition}

\subsection{Intuition}\label{subsec-intuition}

We set the scene, for proving the rank bound of a $\sps(k,d)$ identity, by 
giving a combinatorial/graphical picture to keep in mind. Our circuits consist
of $k$ multiplication terms, and each term
is a product of $d$ linear forms. Think of there being $k$ groups
of $d$ nodes, so each node corresponds to a form and each
group represents a term\footnote{A form that appears many times corresponds to that many 
nodes.}. We will incrementally construct a small
basis for all these forms. This process will be described as some
kind of a \emph{coloring procedure}.

At any intermediate stage, we have a partial basis of forms. These are all linearly
independent, and the corresponding nodes (we will use node and form
interchangeably) are colored \emph{red}. Forms not in the basis
that are linear combinations
of the basis forms (and are therefore in the span of the basis) are colored
\emph{green}. Once all the forms are colored, either green or red, all the red
forms form a basis of \emph{all forms}. The number of red forms
is the rank of the circuit. When we have a partial basis, we carefully
choose some uncolored forms and color them red (add them to the basis). As 
a result, some other forms get ``automatically" colored green (they get added to the
span). We ``pay" only for the red forms, and we would like
to get many green forms for ``free". Note that we are trying
to prove that the rank is $k^{O(1)}\log d$, when the 
total number of forms is $kd$. Roughly speaking, for every
$k^{O(1)}$ forms we color red, we need to show that
the number of green forms will \emph{double}.

So far nothing ingenious has been done. Nonetheless, this image
of coloring forms is very useful to get an intuitive and
clear idea of how the proof works. The main challenge comes in choosing
the right forms to color red. Once that is done, how do we keep an accurate count
on the forms that get colored green? One of the
main conceptual contributions of this work is 
the idea of \emph{matchings}, which aid us in these tasks.
Let us start from a trivial example. Suppose we have
two terms that sum to zero, i.e. $T_1 + T_2 = 0$.
By unique factorization of polynomials, for every
form $\ell \in T_1$, there is a unique form $m \in T_2$
such that $\ell = cm$, where $c \in \FF^*$ (we will
denote this by $\ell \sim m$). By associating
the forms in $T_1$ to those in $T_2$, we create
a \emph{matching} between the forms in these two
groups (or terms). This rather simple observation
is the starting point for the construction of matchings.

Let us now move to $k=3$, so we have 
a simple circuit
$C \equiv T_1 + T_2 + T_3 = 0$.
Therefore, there are no common factors 
in the terms.
To get matchings, we will look at $C$ modulo some
forms in $T_3$. By looking at $C$ modulo various forms
in $T_3$, we reduce the fanin of $C$ and get
many matchings. Then we can deduce structural results
about $C$. Similar ideas were used by Dvir and Shpilka~\cite{DS06} for 
their rank bound.
Taking a form $q \in T_3$, we look at
$C(\mod q)$ which gives $T_1 + T_2 = 0 (\mod q)$.
By unique factorization of polynomials modulo $q$,
we get a \emph{$q$-matching}. Suppose $(\ell, m)$
is an edge in this matching. In terms of the coloring
procedure, this means that if $q$ is colored
and $\ell$ gets colored, then $m$ must also be colored.
At some intermediate stage of the coloring,
let us choose an uncolored form $q \in T_3$.
A key structural lemma
that we will prove is that in the $q$-matching 
(between $T_1$ and $T_2$) any neighbor
of a colored form \emph{must be uncolored}.
This crucially requires the simplicity of $C$.
We will color $q$ red, and thus all neighbors
of the colored forms in $T_1 \cup T_2$ will
be colored green. By coloring $q$ red,
we can double the number of colored forms.
It is the various matchings (combined with
the above property) that allow us to show
an exponential growth in the colored forms
as forms in $T_3$ are colored red.
By continuing this process, we can color 
all forms by coloring at most $O(\log d)$
forms. Quite surprisingly, the above verbal
argument can be formalized easily to prove
that rank of a minimal, simple circuit
with top fanin $3$ is at most $(\log_2 d + 2)$.
For this case of $k=3$, the logarithmic rank bound was there in
a lemma of Dvir and Shpilka~\cite{DS06}, though they did not present the proof idea 
in this form, in particular, their rank bound grew to $(\log d)^2$ for $k=4$.

The major difficulty arises when we try to push these
arguments for higher values of $k$. In essence,
the ideas are the same, but there are many
technical and conceptual issues that arise.
Let us go to $k=4$. The first attempt is
to take a form $q \in T_4$ and look
at $C(\mod q)$ as a fanin $3$ circuit. Can
we now simply apply the above argument
recursively, and cover all the forms
in $T_1 \cup T_2 \cup T_3$? No, the possible 
lack of simplicity in $C(\mod q)$ blocks this simple idea.
It may be the case that $T_1, T_2$ and $T_3$
have no common factors, but once we go
modulo $q$, there could be many common factors!
(For example, let $q = x_1$. Modulo $q$, the forms
$x_1 + x_2$ and $x_2$ would be common factors.)

Instead of doing things recursively (both~\cite{DS06} and~\cite{ks07} used recursive 
arguments), we look at generating matchings iteratively. By performing a careful 
iterative analysis that keeps track
of many relations between the linear forms we achieve a stronger bound for $k > 3$. 
We start with a form $\ell_1 \in T_1$, and look at $C(\mod \ell_1)$.
From $C(\mod \ell_1)$, we remove all common factors. This common factor part we
shall refer to as the \emph{gcd} of $C(\mod \ell_1)$,
the removal of which leaves the \emph{simple}
part of $C(\mod \ell_1)$. Now, we choose
an appropriate form $\ell_2$ from the
simple part, and look at $C(\mod \ell_1, \ell_2)$.
We now choose an $\ell_3$ and so on and so forth.
For each $\ell$ that we choose, we decrease the top
fanin by at least $1$, so we will end up
with a matching modulo the \emph{ideal} $(\ell_1, \ell_2,...,\ell_r)$,
where $r \leq (k-2)$. We call these special
ideals \emph{form ideals} (as they are generated by forms), and the main structures
that we find are matchings modulo form ideals. The
coloring procedure will color the forms in the form
ideal red. Of course, it's not as simple
as the case of $k=3$, since, for one thing, we have to deal
with the simple and gcd parts. Many other problems
arise, but we will explain them as and when we see them.
For now, it suffices to understand the overall picture
and the concept of matchings among the linear forms in $C$. 

We now start by setting some notation and giving
some key definitions.

\subsection{Ideal Matchings}\label{subsec-ideal-mat}

We will use the concept of \emph{ideal matchings} to
develop tools to prove Theorem~\ref{thm-main}. In this
subsection, we provide the necessary definitions and prove
some basic facts about these matchings.

First, we discuss \emph{similarity} between forms and \emph{form ideals}.


\begin{definition} We give several definitions :

\begin{itemize}

\item {\bf [Similar forms]} For any two polynomials $f,g\in R$ we call $f$ {\em similar to} $g$ if there 
is a $c\in\FF^*$ such that $f=cg$. We say $f$ {\em is similar to $g$ mod 
$I$}, for some ideal $I$ of $R$, if there is a $c\in\FF^*$ such that 
$f = cg$ $(\mod I)$. We also denote this by $f \sim g$ $(\mod I)$ or $f$ is $I$-similar to $g$.

\item {\bf [Similar lists]} Let $S_1=(a_1,\ldots,a_d)$ and $S_2=(b_1,\ldots,b_d)$ be two lists of 
linear forms with a bijection $\pi$ between them. $S_1$ and $S_2$ are called 
{\em similar under $\pi$} if for all $i\in[d]$, $a_i$ is similar to $\pi(a_i)$. 
Any two lists of linear forms are called {\em similar} if there exists such a 
$\pi$. Empty lists of linear forms are similar vacuously. For any 
$\ell\in L(R)$ we define the {\em list of forms in $S_1$ similar to $\ell$} as 
the following list (unique upto ordering):
$$simi(\ell,S_1) := (a\in S_1 \mid a\text{ is similar to }\ell)$$
We call {\em $S_1$, $S_2$ coprime lists} if $\forall\ell\in S_1$, 
$\#simi(\ell,S_2)=0$.

\item {\bf [Form-ideal]} A {\em form-ideal} $I$ is the ideal $(I)$ of $R$ generated by some 
nonempty $I\subseteq L(R)$. Note that if $I=\{0\}$ then $a\equiv b (\mod I)$ simply 
means that $a=b$ absolutely. 

\item {\bf [Span $sp(S)$]} For any $S\subseteq L(R)$ we let $sp(S)\subseteq L(R)$ be the 
{\em linear span} of the linear forms in $S$ over the field $\FF$. 

\item {\bf [Orthogonal sets of forms]} Let $S_1,\ldots, S_m$ be sets of linear forms for 
$m\geq2$. We call $S_1,\ldots,S_m$ {\em orthogonal} if for all $m^\prime\in[m-1]$: 
$$sp\big(\bigcup_{j\in[m^\prime]}S_j\big)\cap sp(S_{m^\prime+1}) = \{0\}$$
Similarly, we can define {\em orthogonality of form-ideals} $I_1,\ldots,I_m$. 
\end{itemize}
\end{definition}

We give a few simple facts based on these definitions. It will be helpful
to have these explicitly stated.

\begin{fact}\label{fac-sub-similar}
Let $U, V$ be lists of linear forms and $I$ be a form-ideal. If $U, V$ are 
similar then their sublists $U^\prime:=$ $(\ell\in U\mid\ell\in sp(I))$ and 
$V^\prime:=$ $(\ell\in V\mid\ell\in sp(I))$ are also similar.
\end{fact}
\begin{proof}
If $U, V$ are similar then for some $c\in\FF^*$, $M(V)=c M(U)$. This implies:
$$M(V^\prime)\cdot M(V\setminus V^\prime) =
c M(U^\prime)\cdot M(U\setminus U^\prime)$$
Since elements of $U\setminus U^\prime$ are not in $sp(I)$, 
for any $\ell\in V^\prime$, $\ell$ does not divide $M(U\setminus U^\prime)$.
In other words $M(V^\prime)$ divides $M(U^\prime)$, and vice versa. Thus,
$M(U^\prime), M(V^\prime)$ are similar and hence by unique factorization in 
$R$, lists $U^\prime, V^\prime$ are similar.
\end{proof}

\begin{fact}\label{fac-orth-ideals}
Let $I_1, I_2$ be two orthogonal form-ideals of $R$ and let $D$ be a 
$\sps(k,d)$ circuit such that $L(D)$ has all its linear forms in $sp(I_1)$. 
If $D\equiv0\ (\mod I_2)$ then $D=0$.
\end{fact}
\begin{proof}
As $I_1, I_2$ are orthogonal we can assume $I_1$ to be 
$\{\ell_1,\ldots,\ell_m\}$ and $I_2$ to be 
$\{\ell_1^\prime,\ldots,\ell_{m^\prime}^\prime\}$ where the ordered set $V:=$ 
$\{\ell_1,\ldots,\ell_m$, $\ell_1^\prime,\ldots,\ell_{m^\prime}^\prime\}$ has
$(m+m^\prime)$ linearly independent linear forms. Clearly, there exists an 
invertible linear transformation $\tau$ on $sp(\{x_1,\ldots,x_n\})$ that maps 
the elements of $V$ bijectively, in that order, to $x_1,\ldots,x_{m+m^\prime}$. 
On applying $\tau$ to the equation $D\equiv0\ (\mod I_2)$ we get: 
$$\tau(D)\equiv0\ (\mod x_{m+1},\ldots,x_{m+m^\prime}), \text{~~where }
\tau(D)\in\FF[x_1,\ldots,x_m].$$ 
Obviously, this means that $\tau(D)=0$ which by the invertibility of $\tau$ 
implies $D=0$.
\end{proof}

We now come to the most important definition of this section. We
motivated the notion of \emph{ideal matchings} in the intuition section.
Thinking of two lists of linear forms as two sets of vertices, a matching
between them signifies some linear relationship between the forms
modulo a form-ideal.

\begin{definition} {\bf [Ideal matchings]}
Let $U, V$ be lists of linear forms and $I$ be a form-ideal. An {\em ideal matching 
$\pi$ between $U, V$ by $I$} is a bijection $\pi$ between lists $U, V$ such 
that: for all $\ell\in U$, $\pi(\ell)=c\ell+v$ for some $c\in\FF^*$ and 
$v\in sp(I)$. The matching $\pi$ is called {\em trivial} if $U, V$ are similar.
\end{definition}

Note that $\pi$ being a bijection and $c$ being nonzero 
together imply that $\pi^{-1}$ can also be viewed as a matching between 
$V, U$ by $I$. We will also use the terminology \emph{$I$-matching
between $U$ and $V$} for the above.
Similarly, an {\em $I$-matching $\pi$ between multiplication terms $f, g$} 
is the one that matches $L(f), L(g)$. (For convenience, we will
just say ``matching" instead of ``ideal matching".)

The following is an easy fact about matchings.

\begin{fact}\label{fac-unscramble}
Let $\pi$ be a matching between lists of linear forms $U,V$ by $I$
and let $U^\prime\subseteq U$, $V^\prime\subseteq V$ be similar sublists. Then 
there exists a matching $\pi^\prime$ between $U,V$ by $I$ such that: $U^\prime$, 
$V^\prime$ are similar under $\pi^\prime$.
\end{fact}
\begin{proof}
\comment{
Let $\ell^\prime \in U^\prime$ be such that $\pi(\ell^\prime)=d^\prime\ell^\prime+v^\prime$ (for some 
$d^\prime \in\FF^*$ and $v^\prime \in sp(I)$) is not similar to $\ell^\prime$. 
Since $U^\prime$ and $V^\prime$ are similar,
$\#simi(\ell^\prime,U') = \#simi(\ell^\prime,V')$. Therefore, there must
be some form in $V^\prime$ equal to $\alpha\ell^\prime$ ($\alpha \in\FF^*$) such that $\pi^{-1}$ maps
it to a non-similar form.
Let $\ell \in U$ be the pre-image of this form under $\pi$,
so $\pi(\ell)=$ $d\ell+v=$ $\alpha\ell^\prime$
for some $d\in\FF^*$ and $v\in sp(I)$.}

Let $\ell^\prime\in U^\prime$ be such that $\pi(\ell^\prime)=d^\prime\ell^\prime+v^\prime$ 
(for some $d^\prime\in\FF^*$ and $v^\prime\in sp(I)$) is not in $V^\prime$ or is not similar 
to $\ell^\prime$. As 
$V^\prime$ is similar to $U^\prime$ there exists a form equal to $\alpha\ell^\prime$ in 
$V^\prime$, for some $\alpha\in\FF^*$, and $\pi$ being a matching must be mapping 
some $\ell\in U$ to $\alpha\ell^\prime$ in $V^\prime$. Also from the matching 
condition there must be some $d\in\FF^*$ and $v\in sp(I)$ such
that $\pi(\ell)=$ $d\ell+v=$ $\alpha\ell^\prime$.

Now we define a new matching $\widetilde{\pi}$ by flipping the images of $\ell$ and 
$\ell^\prime$ under $\pi$, i.e., define $\widetilde{\pi}$ to be the same as $\pi$ on 
$U\setminus\{\ell,\ell^\prime\}$ and: 
$\widetilde{\pi}(\ell)\stackrel{V}{:=}\pi(\ell^\prime)$ and 
$\widetilde{\pi}(\ell^\prime)\stackrel{V}{:=}\pi(\ell)$. Note that $\widetilde{\pi}$ inherits the 
bijection property from $\pi$ and it is an $I$-matching because: 
$\widetilde{\pi}(\ell^\prime)=\alpha \ell^\prime$ for $\alpha\in\FF^*$ and more importantly, 

$$\widetilde{\pi}(\ell) = \pi(\ell^\prime) = 
d^\prime \ell^\prime + v^\prime = d^\prime \(\frac{d\ell+v}{\alpha}\)+v^\prime=
\(\frac{dd^\prime}{\alpha}\)\ell+\(\frac{d^\prime v}{\alpha}+v^\prime\)$$

The form $(\frac{d^\prime v}{\alpha}+v^\prime)$ is clearly in $sp(I)$. Thus, we have obtained now 
a matching $\widetilde{\pi}$ between $U, V$ by $I$ such that the $\ell^\prime \in U^\prime$ 
is similar to $\widetilde{\pi}(\ell^\prime)\in V^\prime$.

Note that we increased the number of forms in $U'$ that are
matched to similar forms in $V'$. If we find another form
in $U'$ that is not matched to a similar form in $V'$,
we can just repeat the above process.
We will end up with the desired matching 
$\pi^\prime$ in at most $\#U^\prime$ many iterations.
\end{proof}

We are ready to present the most important lemma of this section. The 
following lemma shows that there cannot be too many matchings 
between two given nonsimilar lists of linear forms. It is at the heart of 
our rank bound proof and the reason for the logarithmic dependence of the 
rank on the degree. It can be considered as an algebraic generalization of 
the combinatorial result used by Dvir \& Shpilka (Corollary 2.9 of 
\cite{DS06}).

\begin{lemma}\label{lem-general-DS}
Let $U, V$ be lists of linear forms each of size $d>0$ and 
$I_1,\ldots,I_r$ be orthogonal form-ideals such that for all $i\in[r]$, 
there is a matching $\pi_i$ between $U, V$ by $I_i$. If $r>(\log_2 d+2)$ 
then $U, V$ are similar lists.
\end{lemma}

Before giving the proof, let us first put it in the context 
of our overall approach. In the sketch that we gave for $k=3$,
at each step, we were generating orthogonal matchings
between two terms. For each orthogonal matchings we got,
we colored one linear form red (added one form to our basis)
and \emph{doubled} the number of green forms (doubled the number of forms 
in the circuit that are in the span of the basis). This showed that
there is a logarithmic-sized basis for all $L(C)$.
If we take the contrapositive of this, we get that there
\emph{cannot} be too many orthogonal matchings between
two (nonsimilar) lists of forms. For dealing
with larger $k$, it will be convenient
to state things in this way.

\begin{proof}
Let $U_1\subseteq U$ be a sublist such that: there exists a sublist 
$V_1\subseteq V$ similar to $U_1$ for which $U^\prime:=U\setminus U_1$ and 
$V^\prime:=V\setminus V_1$ are coprime lists. Let $U^\prime$, $V^\prime$ 
be of size $d^\prime$. If $d^\prime=0$ then $U, V$ are indeed similar and
we are done already. So assume that $d^\prime>0$. By the hypothesis and 
Fact \ref{fac-unscramble}, for all $i\in[r]$, there exists a matching 
$\pi_i^\prime$ between $U,V$ by $I_i$ such that: $U_1$, $V_1$ are similar 
under $\pi_i^\prime$ {\em and} $\pi_i^\prime$ is a matching between 
$U^\prime$, $V^\prime$ by $I_i$. Our subsequent argument will only consider 
the latter property of $\pi_i^\prime$ for all $i\in[r]$.

Intuitively, it is best to think of the various $\pi_i^\prime$s as bipartite matchings.
The graph $G=(U',V',E)$ has vertices labelled with
the respective form. For each $\pi_i^\prime$ and each $\ell \in U'$,
we add an (undirected) edge tagged with $I_i$ 
between $\ell$ and $\pi_i^\prime(\ell)$. There may be many
tagged edges between a pair of vertices\footnote{It can be shown,
using the orthogonality of the $I_i$'s, 
that an edge can have at most {\em two} distinct tags.}.
We call $\pi_i^\prime(\ell)$ the $I_i$-neighbor of $\ell$
(and vice versa, since the edges are undirected).
Abusing notation, we use \emph{vertex} to refer
to a form in $U^\prime \cup V^\prime$.
We will denote $\bigcup_{j \leq i} I_i$ by $J_i$.

We will now show that there cannot be more than 
$(\log_2 d + 2)$ such perfect matchings in $G$. 
The proof is done by following 
an iterative process that has $r$ phases, one for each $I_i$. 
This is essentially the coloring process that we described
earlier. We maintain a partial basis for the forms in $U' \cup V'$
which will be updated iteratively.
This basis is kept in the set $B$. Note that although we only
want to span $U' \cup V'$, we will use forms in the various $I_i$'s
for spanning.

We start with empty $B$ and initialize by adding some $\ell \in U'$ to $B$. 
In the $i$th round, we will add all forms in 
$I_i$ to $B$. All forms of $U' \cup V'$ in $sp(\{\ell\} \cup J_i)$
are now spanned.
We then proceed to the next round.
To introduce some colorful terminology, a \emph{green}
vertex is one that is in the set $sp(B)$ (a form in $(U' \cup V') \cap sp(B)$).
Here is a nice fact : at the end of a round,
the number of green vertices in $U'$ and $V'$ are the same.
Why? All forms of $I_1$ are in $B$, at the end of any round.
Let vertex $v$ be green, so $v \in sp(B)$.
The $I_1$-neighbor of $v$ is a linear combination
of $v$ and $I_1$. Therefore, the neighbor is in $sp(B)$ and is colored green.
This shows that the number of green vertices in $U$ is equal to the number of 
those in $V$.


Let 
$i_0\in[r]$ be the least index such that $\{\ell\}$, $I_1,\ldots,I_{i_0}$ 
are not orthogonal, if it does not exist then set $i_0:=r+1$. Now we have 
the following easy claim.

\begin{claim}\label{clm-lgdplus2}
The sets $\{\ell\}$, $I_1,\ldots,I_{i_0-1}$ are orthogonal and the sets:
$$\{\ell\}\cup J_{i_0}, I_{i_0+1},\ldots,I_r$$
are orthogonal.
\end{claim}
\claimproof{clm-lgdplus2}{
The ideals $\{\ell\}$, $I_1,\ldots,I_{i_0-1}$ are orthogonal by the minimality of 
$i_0$.

As $I_1,\ldots,I_{i_0}$ are orthogonal but $\{\ell\}$, 
$I_1,\ldots,I_{i_0}$ are not orthogonal we deduce that 
$\{\ell\}\in sp(J_{i_0})$. Thus, 
$\{\ell\} \cup sp(J_{i_0})=$ $sp(J_{i_0})$ which is 
orthogonal to the sets $I_{i_0+1},\ldots,I_r$ by the orthogonality of 
$I_1,\ldots,I_r$.
}

We now show that the green vertices double in at least $(r-2)$ many rounds.
\begin{claim} \label{clm-doubling} For $i\not\in \{1,i_0\}$, the number of green
vertices doubles in the $i$th round.
\end{claim}
\claimproof{clm-doubling} {
Let $\ell'$ be a green vertex, say in $U'$,
at the end of the $(i-1)$th round ($B = \{\ell\} \cup J_{i-1}$).
Consider the $I_i$-neighbor of $\ell'$. This
is in $V'$ and is equal to $(c\ell' + v)$ where $c\in\FF^*$ and $v$ 
is a {\em nonzero} element in $sp(I_i)$ (this is because $U', V'$ are coprime).
If this neighbor is green, then $v$ would be a linear combination
of two green forms, implying $v \in sp(B)$. But by Claim~\ref{clm-lgdplus2}, 
$I_i$ is orthogonal to $B$, implying $v\in sp(B)\cap sp(I_i)=\{0\}$ which is a
contradiction. Therefore, the $I_i$-neighbor
of any green vertex is \emph{not} green. On adding $I_i$ to $B$,
all these neighbors will become green. This completes the proof.
}

We started off with one green vertex $\ell$, and $U'$, $V'$ each 
of size $d'$. This doubling can happen at most $\log_2 d'$ times,
implying that $(r - 2) \leq \log_2 d'$.

\end{proof}
\begin{rem}
The bound of $r=\log_2 d+2$ is achievable by 
lists of linear forms inspired by Section \ref{sec-high-rank-eg}. Fix an odd
$s$ and define:
$$U:=\left\{ (b_1x_1+\cdots+b_{s-1}x_{s-1}+x_s)\ \mid \ b_1,\ldots,b_{s-1}\in
\{0,1\}\text{ s.t. } b_1+\cdots+b_{s-1}\text{ is even}\right\}$$ 
$$V:=\left\{ (b_1x_1+\cdots+b_{s-1}x_{s-1}+x_s)\ \mid \ b_1,\ldots,b_{s-1}\in
\{0,1\}\text{ s.t. } b_1+\cdots+b_{s-1}\text{ is odd}\right\}$$
It is easy to see that over rationals, $\#U=\#V=2^{s-2}$ and for all 
$i\in[s-1]$, there is a matching between $U, V$ by $(x_i)$, furthermore, there
is a matching by $(x_1+\cdots+x_{s-1}+2x_s)$. Thus there are $(log_2|U|+2)$ many
orthogonal matchings between these nonsimilar $U, V$; showing that our Lemma is 
tight.
\end{rem}

\subsection{Ordered Matchings and Simple Parts of Circuits}\label{subsec-ordered-mat}

Before we delve into the definitions and proofs, let us motivate
them by an intuitive explanation.

\subsubsection{Intuition}

Our main goal is to deal with the case $k > 3$. 
The overall picture is still the same. We keep updating a partial
basis $S$ for $L(C)$. This process goes through various \emph{rounds},
each round consisting of \emph{iterations}.
At the end of each round, we obtain a form-ideal $I$ that
is orthogonal to $S$. In the first iteration of a round, we start by 
choosing a form $\ell_1$ in $L(T_1)$ that
is not in $sp(S)$, and adding it to $I$. We look at $C(\mod \ell_1)$ in the next
iteration, which is obviously
an identity, and try to repeat this step. The top fan-in has gone
down by at least one, or in other words, some multiplication terms have become
identically zero $(\mod \ell_1)$.
We will say that the other terms have \emph{survived}.
The major obstacle to proceeding is that our circuit is not simple 
any more, because there \emph{can} be common factors among multiplication terms 
modulo $\ell_1$.
Note how this seems to be a difficulty, since it appears
that our matchings will not give us a proper
handle on these common factors.
Suppose that form $v$ is now a common factor.
That means, in every surviving term, there is
a form that is $v$ modulo $\ell_1$. So these forms
can be $\ell_1$-matched to each other! We have converted
the obstacle into some kind of a partial matching, which
we can hopefully exploit.

Let us go back to $C(\mod \ell_1)$. Let us remove
all common factors from this circuit. This stripped
down identity circuit is the \emph{simple} part, denoted
by $sim(C\mod \ell_1)$. The removed portion, called the $gcd$ part,
is referred to as $gcd(C\mod \ell_1)$. By the above observation,
the $gcd$ part has $\ell_1$-matchings. A key observation
is that all the forms in the $gcd$ part are \emph{not}
similar to $\ell_1$. This is because we were only
looking at nonzero terms in $C(\mod \ell_1)$.
Having (somewhat) dealt with $gcd(C\mod \ell_1)$
by finding $I$-matchings, let us focus on the smaller circuit $sim(C\mod \ell_1)$

We try to find an $\ell_2 \in L(sim(C\mod \ell_1))$
that is not in $sp(S \cup \{\ell_1\})$. Suppose we
can find such an $\ell_2$. Then, we add $\ell_2$ to $I$
and proceed to the next iteration. In a given iteration,
we start with a form-ideal $I$, and a circuit $sim(C\mod I)$.
We find a form $\ell \in L(sim(C\mod I)) \backslash sp(S \cup I)$.
We add $\ell$ to $I$ (for convenience, let us set $I' = I \cup \{\ell\}$)
and look at the $C(\mod I')$.
We now have new terms in the $gcd$ part, which we can match
through $I'$-matchings. As we observed earlier, all the terms
that have forms in $I'$ are removed, so the
terms we match here are all nonzero modulo $I'$.
We remove the $gcd$ part to get $sim(C\mod I')$,
and go to the next iteration with $I'$ as the new $I$.
When does this stop? If there is no $\ell$
in $L(sim(C\mod I)) \backslash sp(S \cup I)$,
then this means that all of $L(sim(C\mod I))$
is in our current span. So we happily stop here
with all the matchings obtained from the $gcd$ parts.
Also, if the fan-in reaches $2$, then we can imagine
that the whole circuit is itself in the $gcd$ portion.
At each iteration, the fan-in goes down by at least one, so we can have
at most $(k-2)$ iterations in a round, hence the $I$ in any round is generated by
at most $(k-2)$ forms.
When we finish a round obtaining an ideal $I$, there are some multiplication terms 
in $C$ that are nonzero modulo $I$ {\em after} 
the {\em gcd} parts in the various iterations are removed from these terms. 
These we shall refer to as constituting the 
\emph{blocking subset} of $[k]$, for that round.

The way we prove rank bounds
is by invoking Lemma~\ref{lem-general-DS}. 
Each round constructs a new orthogonal
form ideal. At the end of a round, we
have a set $S$, which is a partial basis. 
If $S$ does not cover all of $L(C)$, then we use the above process
(of iterations) to generate a form-ideal $I$ orthogonal to $S$.
Consider two terms $T_a$ and $T_b$ that survive this
process (mod $I$). At each stage, when we add a form to $I$,
we remove forms from $T_a$ and $T_b$, $I$-matching them.
When we stop with our form-ideal $I$, we can
think of $T_a$ and $T_b$ as split into two parts :
one having forms from $sp(S \cup I)$, and the other
which is $I$-matched. For each orthogonal form-ideal we generate,
we match subsets of terms. We use Lemma~\ref{lem-general-DS}
to tell us that we cannot have too many such form-ideals,
which leads to the rank bound.

\subsubsection{Definitions}

We start with looking at the particular kind of matchings that
we get. Take two terms $T_a$ and $T_b$ that survive a round,
where we find the form-ideal $I$ generated by $\{\ell_1,\ell_2,\cdots,\ell_r\}$.
At the end of the first iteration, we add $\ell_1$ to $I$.
No form in $L(T_a) \cup L(T_b)$ can be $0(\mod \ell_1)$.
We match some forms in $T_a$ to $T_b$ via $\ell_1$-matchings.
They are removed, and then we proceed to the next iteration.
We now match some forms via $sp(\{\ell_1,\ell_2\})$ matchings
and none of these forms are in this span. So in each iteration,
the forms that are matched (and then removed) are non-zero
mod the partial $I$ obtained by that iteration.
We formalize this as an \emph{ordered matching}.

\begin{definition} \label{def-ordered} {\bf [Ordered matching]}
Let $U, V$ be lists of linear forms and an ordered set 
$I=\{v_1,\ldots,v_i\}$ be a form-ideal having $i\geq1$ linearly independent 
linear forms. A matching $\pi$ between $U, V$ by $I$ is called an {\em ordered 
$I$-matching} if : 

Let $v_0$ be zero. For all $\ell\in U$, $\pi(\ell)=$ $(c\ell+w)$ where $c\in\FF^*$,
and $w\in sp(v_0,\ldots,v_j)$ for some $j$ satisfying $\ell \notin sp(v_0,\ldots,v_j)$.
\end{definition}

We add the zero element $v_0$, just to deal with similar forms in $U$ and $V$.
Note that the inverse bijection $\pi^{-1}$ is also an ordered matching 
between $V, U$ by $I$. It is also easy to see that if $\pi_1$ and $\pi_2$ 
are ordered matchings between lists $U_1, V_1$ and lists $U_2, V_2$ 
respectively by the same ordered form-ideal $I$ then their {\em disjoint 
union}, $\pi_1\sqcup\pi_2$, is an ordered matching between lists
$U_1\cup U_2$, $V_1\cup V_2$ by $I$.

We will stick to the notation in Definition~\ref{def-ordered}.
For convenience, let $sp_j:=sp(v_0,\cdots,v_j)$.
Let $\pi(\ell) = d\ell + w$, where $w \in sp_j$ but $\ell\not\in sp_j$ then 
the constant $d$ is unique. If there
were two such different constants, say $d$ and $d^\prime$, then both $(\pi(\ell)-d\ell)$ 
and $(\pi(\ell)-d^\prime\ell)$ would be in $sp_j$ implying that 
$(d-d^\prime)\ell\in sp_j$. That contradicts $\ell\not\in sp_j$.
Thus for a fixed $\ell$ and an ordered matching $\pi$, $d$ is uniquely determined. Keeping
the notation above, we can well define :


\begin{definition} {\bf [Scaling factor]} The scaling factor
of an ordered matching $\pi$ between $U$ and $V$ is denoted by $sc(\pi)$.
For each $\ell \in U$, let $d_\ell$ be the unique constant
such that $\pi(\ell) = d_\ell \ell + w$, where $w \in sp_j$ but $\ell\not\in sp_j$.
Then $sc(\pi) := \prod_{\ell\in U} d_\ell$.
For empty $U$, $sc(\pi)$ is set to be $1$.	
\end{definition}


\begin{definition} {\bf [Subcircuits and regular circuits]} For non-empty 
$Q \subseteq [k]$, the {\em subcircuit} $C_Q$ of a $\sps(k,d)$ circuit $C$ is the sum 
$\sum_{j\in Q}T_j$. For a 
form-ideal $I$ we call {\em $C_Q$ regular mod $I$} if $\forall q\in Q$, 
$T_q \not\equiv0\ (\mod I)$. 
We will denote the constant factor in the multiplication term $T_q$ by $\alpha_q\in\FF^*$, 
thus $T_q = \alpha_q M(L(T_q))$.
\end{definition}

We are now ready to define the $gcd$ and $sim$ parts
of a subcircuit. Although the ideas are quite simple
and intuitive, we have to be careful in dealing
with constant factors. Much of this notation has been
introduced for rigorous definitions.
Take a subcircuit $C_Q$ that is regular mod $I$ as well as an identity mod $I$. 
A maximal list of forms, say $U$, that divides $T_q$, for all $q \in Q$,
is called the $gcd$ of $C_Q (\mod I)$. In every $T_q$, there
is a list $U_q$ of forms that are $I$-similar to $U$.
Therefore, we have $I$-matchings between $U$ and $U_q$.
This is the \emph{gcd data of $C_Q$ modulo $I$}, and
represents that various matchings that we will later
exploit. If we remove $U_q$ from each $T_q$, then
(by accounting for constants carefully) we get
a simple $(\mod I)$ identity, the $sim$ part of $C_Q(\mod I)$.
We formalize this below.
\\

Let $C_Q$ be regular modulo $I$.
Fix a 
$q_1$ in $Q$. Let $U$ be a maximal sublist of $L(T_{q_1})$ such that 
$M(U)$ divides $T_q$ modulo $I$ for all $q\in Q$. Since $R/I$ is 
isomorphic to a polynomial ring, the nonconstant polynomials in $R/I$ 
satisfy unique factorization property, i.e. any polynomial in $R$ that 
is nonconstant modulo $I$ uniquely factors modulo the ideal $(I)$ into 
polynomials irreducible modulo $I$. Since $C_Q$ is regular modulo 
$I$ {\em and} $U\subseteq L(T_{q_1})$ is a maximal list such that 
$\forall q\in Q$, $M(U)\mid T_q (\mod I)$:
\begin{itemize}
\item $M(U)$ is a gcd of the polynomials $\{T_q\mid q\in Q\}$ modulo 
the ideal $(I)$.
\item For all $q\in Q$, there exists a sublist $U_q\subseteq L(T_q)$
and a $c_q\in\FF^*$ such that $M(U_q)\equiv c_q\cdot M(U)$ 
$(\mod I)$. By unique factorization in $R/I$ and regularity of $C_Q$ mod 
$I$ this gives an ordered matching $\pi_q$ between $U, U_q$ by $I$. Also, 
by the definition of scaling factor of a matching, $\pi_q$ satisfies: 
$\forall q\in Q$, $M(U_q)\equiv sc(\pi_q)\cdot M(U)$ $(\mod I)$.
\end{itemize}
Note that given $C_Q$ and $I$ there are many possibilities to 
choose the lists $U$ and $\{U_q\mid q\in Q\}$ but they are all 
uniquely determined upto similarity modulo the ideal $(I)$ and that 
will be good enough for our purposes. So we choose them in some 
way, say the lexicographically smallest one unless specified otherwise, 
and define the gcd data.
Using the gcd data of $C_Q$ mod $I$ we can extract out a smaller 
circuit from $C_Q$ which we call the simple part.

\begin{definition} {\bf [gcd and sim parts]} 
The {\em gcd data of $C_Q$ modulo $I$} is the following set 
of $\#Q$ matchings:
\begin{equation}\label{eqn-gcd-data}
	\vec{gcd}(C_Q \mod I) := \left\{(\pi_q,U,U_q)\mid q\in Q \right\}
\end{equation}
The {\em gcd of $C_Q(\mod I)$} is just $gcd(C_Q\mod I):=M(U)$.
The {\em simple part of $C_Q$ mod 
$I$} is the circuit:
$$sim(C_Q\mod I) := \sum_{q\in Q}sc(\pi_q)\alpha_q\cdot M(L(T_q)\setminus U_q) $$
\end{definition}

Before a round, we have a partial basis $S$.
At the end of a round, we produce a form-ideal $I$ that is orthogonal
to $S$. We call this a \emph{useful ideal}.
Let $Q \subset [k]$ be such that all $T_q$, $q \in Q$ survive (mod $I$).
This is called the \emph{blocking subset}.
For each such $q$, there are a list of forms $V_q \subset L(T_q)$
that are mutually matched via ordered $I$-matchings (these are really a collection 
of $gcd$ datas). This is called the \emph{matching data}.
Even after we remove $V_q$ from each term $T_q$ (carefully accounting for constants, 
as explained above), we still have an identity mod $I$. All forms of this identity are 
in $sp(S \cup I)\setminus sp(I)$, since we assume that the round has ended.
Furthermore by rearranging linear forms, all $V_q$'s can be made disjoint to 
$sp(S \cup I) \backslash sp(I)$. Therefore this round partitions the $L(T_q)$ into $V_q$ 
and $L(T_q)\cap(sp(S \cup I) \backslash sp(I))$ (for all $q\in Q$).
These end-of-a-round properties are formalized by the following definition.

\begin{definition} {\bf [Useful ideals, blocking subsets, and matching data]}
Let $C = \sum_{j \leq k} T_j$, $T_j=\alpha_j$ $M(L(T_j))$.
The set $S\subseteq L(R)$ and $I$ is an ordered form-ideal orthogonal to $S$. 
We call $I$ {\em useful in $C$ wrt $S$} if $\exists Q\subset[k]$, $1<\#Q<k$ 
with the following properties :

For all $q\in Q$, let $V_q$ be $L(T_q) \backslash (sp(S\cup I)\setminus sp(I))$.
(Therefore, $L(T_q)\setminus V_q \subset sp(S\cup I)\setminus sp(I)$.)

\begin{itemize}
\item There exists a list of linear forms $V$ such that for all $q\in Q$,
there is an ordered $I$-matching $\tau_q$ between $V, V_q$.
\item The circuit $\sum_{q\in Q}sc(\tau_q)\alpha_q\cdot M(L(T_q)\setminus V_q)$ is 
a regular identity modulo $I$.
\end{itemize}
Such a $Q$ we call a {\em blocking subset} of $C,S,I$. By {\em matching 
data} of $C,S,I,Q$ we will mean the set:
$$mdata(C,S,I,Q):=\left\{(\tau_q,V,V_{q})\mid q \in Q\right\}$$

We will call $mdata(C,S,I,Q)$ {\em trivial} if the 
lists $V_q$, $q\in Q$, are all mutually similar. 
\end{definition} 

From the matching data, we will exploit the fact that 
for each pair $q_1, q_2 \in Q$, there is an
ordered $I$-matching between $V_{q_1}$ and $V_{q_2}$.
Nonetheless, we will represent these $\#Q$ matchings
via $V$ because it will be more convenient to deal
with the intermediate $gcd$ parts while we are building $I$.

\subsubsection{Basic facts}

In this subsection, we prove some basic facts about ordered matchings,
scaling factors and $gcd$ and $sim$ parts of a circuit. These facts are
not difficult to prove, but it will be helpful later to have them.

The following two properties are immediate from the definition of scaling
factor.
\begin{fact}\label{fac-inv-and-union}
Let $\pi_1$ and $\pi_2$ be ordered $I$-matchings between lists $U_1, V_1$ 
and lists $U_2, V_2$ respectively.
Then $sc(\pi_1^{-1})=sc(\pi_1)^{-1}$ and $sc(\pi_1\sqcup\pi_2)=$ 
$sc(\pi_1)\cdot sc(\pi_2)$.
\end{fact}

Thus, ordered matchings have inverses, have a union and the following fact 
shows that they can also be composed.
\begin{fact}\label{fac-compose}
Let $\pi_1$ and $\pi_2$ be ordered matchings between $U_1, V$ and 
$V, U_2$ respectively by the same ordered form-ideal $I=\{v_1,\ldots,v_i\}$.
Then the naturally defined composite matching $\pi_2\pi_1$ is also an
ordered matching between $U_1, U_2$ by $I$. Furthermore,  
$sc(\pi_2\pi_1)=$ $sc(\pi_1)\cdot sc(\pi_2)$.
\end{fact}
\begin{proof}
Consider a linear form $\ell\in U_1$. There exists $c_1\in\FF^*$ and 
$\alpha_1\in sp_{j_1}, \ell \notin sp_{j_1}$ such that $\pi_1(\ell)=c_1\ell+\alpha_1$.
Also, there exists $c_2\in\FF^*$ and $\alpha_2\in sp_{j_2}$, $\pi_1(\ell) \notin sp_{j_2}$
such that $\pi_2(\pi_1(\ell))=$ $c_2(c_1\ell+\alpha_1)+\alpha_2$.
Let $j = \max \{j_1,j_2\}$. Obviously, $(c_2\alpha_1 + \alpha_2) \in sp_j$.
If $\ell\in sp_j$ then as $\ell\not\in sp_{j_1}$ we deduce that $j=j_2 > j_1$, thus 
$\ell \in sp_{j_2}$, implying $\pi_1(\ell)=c_1\ell + \alpha_1 \in sp_{j_2}$, which is a 
contradiction. Therefore, $\ell \notin sp_j$. This proves that the composite bijection 
$\pi_2 \pi_1$ is an ordered matching.

The contribution from the image of $\ell\in U_1$ to 
$sc(\pi_2\pi_1)$ is $c_1c_2$ while the corresponding contributions of 
$\ell\in U_1$ to $sc(\pi_1)$ is $c_1$ and of $\pi_1(\ell)\in V$ to $sc(\pi_2)$ 
was $c_2$. Thus, $sc(\pi_2\pi_1)=$ $sc(\pi_1)\cdot sc(\pi_2)$.
\end{proof}

The scaling factor nicely characterizes the ratio of $M(U)$ and $M(V)$ when 
$U, V$ are similar.

\begin{fact}\label{fac-scaling1}
Let $\pi$ be an ordered matching between lists $U, V$ of linear forms, by
an ordered form-ideal $I=\{v_1,\ldots,v_i\}$. If $\pi$ is trivial then
$M(V)=sc(\pi)\cdot M(U)$. Thus all the ordered matchings, between a given pair of
similar lists, have the same scaling factor.
\end{fact}
\begin{proof}
The proof idea is identical to the one seen in Fact \ref{fac-unscramble}.

Let $\ell\in U$ be such that $\pi(\ell)=d\ell+v$ is not similar to 
$\ell$, where $d\in\FF^*$, $v\in sp_j$ and 
$\ell \notin sp_j$. 
Since $V$ is similar to $U$ there exists a form equal to $c\ell$ in $V$, 
for some $c\in\FF^*$. As $\pi$ is an ordered matching,
it must be mapping some $\ell^\prime\in U$ to $c\ell$ in $V$, satisfying:  
$\pi(\ell^\prime)=$ 
$d^\prime\ell^\prime+v^\prime=$ $c\ell$,
where $d^\prime\in\FF^*$,
$v^\prime\in sp_{j'}$,
and $\ell^\prime \notin sp_{j'}$.

Now we define a new matching $\widetilde{\pi}$ by flipping the images of $\ell$ 
and $\ell^\prime$ under $\pi$, i.e., define $\widetilde{\pi}$ to be the same as 
$\pi$ on $U\setminus\{\ell,\ell^\prime\}$ and: 
$\widetilde{\pi}(\ell)\stackrel{V}{:=}\pi(\ell^\prime)$ and 
$\widetilde{\pi}(\ell^\prime)\stackrel{V}{:=}\pi(\ell)$. The
matching $\widetilde{\pi}$ is an ordered matching 
because: $\widetilde{\pi}(\ell)=c\ell$ for $c\in\FF^*$ and more importantly 
$\widetilde{\pi}(\ell^\prime)=$ $d\ell+v=$ 
$d(\frac{d^\prime\ell^\prime+v^\prime}{c})+v=$ 
$(\frac{dd^\prime}{c})\ell^\prime+(\frac{dv^\prime}{c}+v)$.
Let $j^* := \max\{j,j^\prime\}$. Obviously, 
$(\frac{dv^\prime}{c}+v) \in sp_{j^*}$. If $j^* = j^\prime$,
we are done, because we already know that $\ell^\prime \notin sp_{j^\prime}$.
If $j^* = j$ and $\ell^\prime \in sp_{j}$, then $c\ell =$
$d^\prime \ell^\prime + v'$ is in $sp_j$ (contradiction).

We have obtained now an 
ordered matching $\widetilde{\pi}$ between $U, V$ by $I$ where
the number of forms mapped to a similar form has strictly increased.
Observe that $sc(\pi)$ had a unique 
contribution of $d$, $d^\prime$ from the images of $\ell$, $\ell^\prime$ 
respectively while $sc(\widetilde{\pi})$ has a corresponding contribution of $c$, 
$(\frac{dd^\prime}{c})$. On all the other elements of $U$, $\widetilde{\pi}$ is 
the same as $\pi$. Thus, we have that $sc(\widetilde{\pi})=sc(\pi)$. 

The above process will yield an ordered matching $\pi^\prime$ in at 
most $\#U$ many iterations, such that $U, V$ are similar under 
$\pi^\prime$ and $sc(\pi^\prime)=sc(\pi)$. But this means that, for all 
$\ell\in U$, $\pi^\prime(\ell)=\lambda\ell$, for some $\lambda\in\FF^*$. 
By definition the contribution by $\ell$ to $sc(\pi^\prime)$ would be then
$\lambda$. This clearly implies that $M(V)=sc(\pi^\prime)\cdot M(U)$
and finally $M(V)=sc(\pi)\cdot M(U)$.
\end{proof}

We move on to facts about the $gcd$ and $sim$ parts
of a circuit.

\begin{fact}\label{fac-simple}
If $C_Q$ is a regular mod $I$ subcircuit of $C$ then:
$$C_Q \equiv gcd(C_Q\mod I)\cdot sim(C_Q\mod I) \ (\mod I)$$
Additionally, if $C_Q$ is an identity modulo $I$ then 
$sim(C_Q\mod I)$ is a simple identity modulo $I$. 
\end{fact}
\begin{proof}
Recall that $C_Q=\sum_{q\in Q}T_q$ and
the $gcd$ data $\vec{gcd}(C_Q \mod I)$ is $\left\{(\pi_q,U,U_q)\mid q\in Q \right\}$.
Now $T_q=$
$\alpha_q M(U_q)\cdot M(L(T_q)\setminus U_q)$ and $M(U_q)\equiv$ $sc(\pi_q)\cdot M(U)$
$(\mod I)$,
where $M(U)$ is $gcd(C_Q \mod I)$. Thus, 
\begin{align*}
C_Q\ &\equiv\ \sum_{q\in Q} \alpha_q sc(\pi_q)M(U)\cdot 
M(L(T_q)\setminus U_q)\ (\mod I)\\
&\equiv\ gcd(C_Q\mod I)\cdot sim(C_Q\mod I) \ (\mod I)
\end{align*}
This proves the first part. Assume now that $C_Q\equiv0 (\mod I)$
which means $sim(C_Q\mod I)\equiv0 (\mod I)$. If it is not a simple
identity mod $I$, then there is an $\ell^\prime\in L(sim(C_Q\mod I))$ 
such that, $\forall q\in Q$, $\ell^\prime\mid$ 
$M(L(T_q)\setminus U_q)$ $\mod I$. 
%
Then, $M(U)$ cannot be the gcd of the polynomials 
$\{T_q\mid q\in Q\}$ modulo the ideal $(I)$ (contradiction). 
\end{proof}

When $I=\{0\}$ we write $\vec{gcd}(C_Q)$, $gcd(C_Q)$ and 
$sim(C_Q)$ instead of $\vec{gcd}(C_Q\mod I)$, $gcd(C_Q\mod I)$ 
and $sim(C_Q\mod I)$ respectively. We collect here some properties 
of $sim(C_Q)$ that would be directly useful in our rank bound proof.

\begin{fact}\label{fac-sim-1}
Let $\ell\in L(R)^*$ and $C_Q$ be a subcircuit of $C$. Then 
$\#simi(\ell,L(sim(C_Q)))>0$ iff $\exists q_1, q_2\in Q$ such that
$\#simi(\ell,L(T_{q_1}))\not=$ $\#simi(\ell,L(T_{q_2}))$.
\end{fact}
\begin{proof}
Note that $\#simi(\ell,L(T_q))$ is the highest power of $\ell$ that divides
$T_q$. Thus, if $\#simi(\ell,L(T_q))$ is the same, say $r$, for all $q\in Q$ 
then the highest power of $\ell$ dividing $gcd(C_Q)$ is also $r$ implying 
that for all $q\in Q$, the polynomial $\frac{T_q}{gcd(C_Q)}$ is coprime to 
$\ell$. By definition of the simple part of $C_Q$ this means that 
$\#simi(\ell,L(sim(C_Q)))=0$.

Conversely, if for an $\ell\in L(R)^*$, $\exists q_1, q_2\in Q$ such that 
$\#simi(\ell,L(T_{q_1}))>$ $\#simi(\ell$, $L(T_{q_2}))$ then it is easy to see 
that $\frac{T_{q_1}}{gcd(C_Q)}$ cannot be coprime to $\ell$. 
This implies that
$\#simi(\ell,L(sim(C_Q))) >0$.
\end{proof}

\begin{fact}\label{fac-sim-2}
Let $S\subseteq L(R)$ and $Q_2\subseteq Q_1\subseteq[k]$. If 
$L(sim(C_{Q_1}))$ has all its linear forms in $sp(S)$, then all the linear 
forms in $L(sim(C_{Q_2}))$ are also in $sp(S)$.
\end{fact}
\begin{proof}
For an arbitrary $\ell\in L(sim(C_{Q_2}))$, by Fact \ref{fac-sim-1}, there are 
$q_1, q_2\in Q_2$ such that $\#simi(\ell,L(T_{q_1}))$ $\not=$ 
$\#simi(\ell,L(T_{q_2}))$. As $q_1, q_2\in Q_1$, we can again apply Fact 
\ref{fac-sim-1} to deduce that $\#simi(\ell$, $L(sim(C_{Q_1})))>0$.
Therefore $\ell\in sp(S)$.
\end{proof}

\begin{fact}\label{fac-sim-3}
Let $S\subseteq L(R)$ and $Q_1, Q_2\subseteq[k]$ such that 
$Q_1\cap Q_2\not=\phi$. If $L(sim(C_{Q_1}))$ and $L(sim(C_{Q_2}))$ have 
all their linear forms in $sp(S)$ then all the linear forms in 
$L(sim(C_{Q_1\cup Q_2}))$ are also in $sp(S)$.
\end{fact}
\begin{proof}
Take $q_0\in Q_1\cap Q_2$ and an arbitrary $\ell\in L(sim(C_{Q_1\cup Q_2}))$.
By Fact \ref{fac-sim-1}, there are $q_1, q_2\in Q_1\cup Q_2$ such that 
$\#simi(\ell,L(T_{q_1}))\not=$ $\#simi(\ell,L(T_{q_2}))$. 

If $q_1, q_2$ are in the same set (wlog, in $Q_1$), then Fact 
\ref{fac-sim-1} tells us that $\#simi(\ell$, $L(sim(C_{Q_1})))>0$,
trivially implying that $\ell\in sp(S)$.
Now assume wlog that $q_1\in Q_1, q_2\in Q_2$. 
For some $i \in \{1,2\}$, $\#simi(\ell,L(T_{q_0}))\not=$ $\#simi(\ell,L(T_{q_i}))$.
Therefore, by Fact \ref{fac-sim-1}, $\ell \in sp(S)$.
\end{proof}

\subsection{Getting Useful Form-ideals}\label{subsec-useful}

Given a set $S$ that does not span all of $L(C)$, we 
can find a form-ideal that is useful wrt $S$.
As we mentioned earlier, in a \emph{round} we start with $S$,
and end up with a useful $I$ through various iterations.
We will formally describe this process below.

An iteration starts with a partial $I$, and a simple regular
identity $E$ in the ring $R / I$, which has multiplication terms with indices in $[k]$.
At least one of the forms in $E$ is \emph{not}
in $sp(S \cup I)$. At the beginning of the first iteration,
$E$ is set to $C$ and $I$ is $\{0\}$.

\begin{center}
\fbox{\begin{minipage}{\columnwidth} {\sc A single iteration}
\begin{enumerate}
	\item Let $\ell$ be a form in $E$ that is not in $sp(S \cup I)$.
	\item Add $\ell$ to $I$.
	\item Consider $E$ modulo $I$ and let $Q$ be the subset of indices of nonzero multiplication terms.
	\item Let $U$ be the $gcd$ of $E (\mod I)$, and let the gcd data
	be $\vec{gcd} = \left\{(\pi_q,U,U_q)\mid q\in Q \right\}$.
	\item If the fanin, $|Q|$, of $E(\mod I)$ is $2$, stop the round. 
	\item If all forms in $sim(E(\mod I))$ are contained in $sp(S \cup I)$, stop
	the round. 
	Otherwise, set $E$ to be $sim(E(\mod I))$ and go to the next iteration.
\end{enumerate}
\end{minipage}}
\end{center}

\begin{lemma}\label{lem-useful-ideal}
Let $C$ be a simple $\sps(k,d)$ identity in $R$.
Suppose $S\subseteq L(R)$ 
and $L(C) \backslash sp(S)$ is non-empty. Then there is a 
form-ideal $I$ useful in $C$ wrt $S$.
\end{lemma}
\begin{proof}
As discussed before in the intuition, we generate $I$ in one round and 
the proof will be done by induction on the number of iterations in this round. 
For convenience, we set
the end of the zero iteration to be the beginning
of the round. We will prove the following claim: 

\begin{claim} \label{clm-iteration} Consider the end of some iteration. There exists
a list $V$ of forms such that : for all $q$ in the current $Q$, there
is a list $V_q \subseteq L(T_q)$ that has an ordered $I$-matching
to $V$. Furthermore, $M(L(T_q) \backslash V_q)$
is similar to the term indexed by $q$ in $sim(E(\mod I))$.
\end{claim}

\claimproof{clm-iteration} {This is proven by induction on the iterations. At the end
of the zero iteration, $E$ is just $C$ and $I=\{0\}$. By the
simplicity of $C$, $sim(E(\mod I))$ is just $C$, and $Q=[k]$. So all the $V_q$'s can be 
taken just empty. 

Now, suppose that at the end of the $i$th iteration,
we have an ordered $I$-matching from $V_q$ to $V$
for all $q$ in the current $Q$. In the $(i+1)$th iteration we will denote 
by $I'$ the set $I \cup \{\ell\}$, $E' = sim(E(\mod I))$, and
$Q'\subset Q$ the subset of indices of non-zero terms
in $E'$ modulo $I'$. For a $q \in Q'$, we have a list $V_{q}\subseteq L(T_q)$ and an 
ordered $I$-matching $\tau_{q}$ between $V, V_q$. All forms of $T_{q}$ not in $V_{q}$
are in $E'$. Now consider the $I'$-matching $\pi_q$ between $U, U_{q}$ obtained in this 
iteration. No forms in these can be in $sp(I')$, since $U$ is $gcd(E'(\mod I'))$
and $q \in Q'$. Therefore, $\pi_q$ is an ordered matching.
We can take the disjoint union of these matchings to get
an ordered $I'$-matching $\tau_q \sqcup \pi_q$ between $V \cup U$ and $V_q \cup U_q$.
All forms in $L(T_q) \backslash (V_q \cup U_q)$ are in the $q$th term of 
$sim(E'(\mod I'))$. This completes the proof of the claim.}

The number of iterations in a round is at most
$(k-2)$. This is because after each iteration, the fanin
of the circuit $E$ goes down by at least $1$.
Therefore, there must be a last iteration (signifying the end of the round).
Consider the end of the last iteration. 
If the fanin $|Q|$ of $E(\mod I)$ is $2$,
then by unique factorization, $sim(E(\mod I))$ is empty.
So, all the forms in $sim(E(\mod I))$ are in $sp(S \cup I)$, at the
end of a round. By the previous claim, there is a list $V$ such that 
for every surviving $q \in Q$, there is a sublist $V_q \subseteq L(T_q)$
and an ordered $I$-matching $\tau_q$ between $V$ and $V_q$.
By Fact~\ref{fac-simple}, we have that $E(\mod I)$ is 
$\sum_{q\in Q}sc(\tau_q)\alpha_q\cdot M(L(T_q)\setminus V_q)$ and is an identity (in $R / I$).

Let $V'_q := V_q \setminus (sp(S \cup I) \backslash sp(I))$ (similarly, define $V'$).
Note that $\tau_q$ induces a matching $\tau'_q$ between $V'$ and $V'_q$.
Furthermore, $\sum_{q\in Q}sc(\tau'_q)\alpha_q\cdot M(L(T_q)\setminus V'_q)$
is a multiple of $E(\mod I)$ and is regular (each term in the above
sum is non-zero mod $I$). Thus, form-ideal $I$ is useful in $C$ wrt $S$.
\end{proof}

To prove a rank bound for minimal and simple $\sps(k,d)$ identity $C$, our plan 
is to start with $S=\phi$ and expand it round-by-round by adding the forms of 
a form-ideal, useful in $C$ wrt $S$, to the current $S$. Trivially, such a 
process has to stop in at most $kd$ iterations (over all rounds) but we intend 
to show that it actually ends up, covering all the forms in $L(C)$, in a much faster way. To 
formalize this process we need the notion of a {\em chain of form-ideals}. 
This is just a concise representation of the matchings that we get
from the various rounds.

\begin{definition} {\bf [Chain of form-ideals]}
Let $C$ be a $\sps(k,d)$ circuit. We define a {\em chain of form-ideals for $C$} to 
be the ordered set $\calT:=$ $\{(C,S_1,I_1,Q_1),\ldots,(C,S_m,I_m,Q_m)\}$ where,
\begin{itemize}
\item For all $i\in[m]$, $S_i\subseteq L(R)$, $I_i$ is a form-ideal orthogonal to 
$S_i$ and $Q_i\subseteq[k]$ .
\item $S_1=\phi$ and for all $2\leq i\leq m$, $S_i=S_{i-1}\cup I_{i-1}$.
\item For all $i\in[m]$, $I_i$ is useful in $C$ wrt $S_i$.
\item For all $i\in[m]$, $Q_i$ is a blocking subset of $C,S_i,I_i$.
\end{itemize}
We will use $sp(\calT)$ to mean $sp(S_m\cup I_m)$ and $\#\calT$ to denote $m$,
the {\em length of $\calT$}. The chain $\cal T$ is \emph{maximal}
if $L(C) \subseteq sp(\cal T)$.
\end{definition}

Note that by Lemma~\ref{lem-useful-ideal}, if a chain $\cal T$ of length $m$ is not maximal,
then we can find a form-ideal $I_{m+1}$ that is useful wrt $S_m \cup I_m$. This
allows us to add a new $(C,S_{m+1},I_{m+1}$, $Q_{m+1})$ to this chain.
It is easy to construct a maximal chain for $C$, and the length
of this can be used to bound the rank:

\begin{fact}\label{fac-maximal-chain}
Let $C$ be a simple $\sps(k,d)$ identity. Then there exists a maximal chain of 
form-ideals $\calT$ for $C$. The rank of $C$ is at most $(k-2) (\#{\cal T})$.
\end{fact}
\begin{proof}
We start with $S_1=\phi$ and an $\ell\in L(C)$. By Lemma 
\ref{lem-useful-ideal} there is a form-ideal $I_1$ (containing $\ell$) 
useful in $C$ wrt $S_1$ with blocking subset, say, $Q_1$. So we have a chain 
of form-ideals $\{(C,S_1,I_1,Q_1)\}$ to start with. Now if $L(C)$ has all its 
elements in $sp(S_1\cup I_1)$ then the chain cannot be extended any further 
and we are done. Otherwise, we can again apply Lemma $\ref{lem-useful-ideal}$ to 
get a form-ideal $I_2$ useful in $C$ wrt $S_2:=S_1\cup I_1$ with blocking 
subset, say, $Q_2$. Thus, we have a longer chain of form-ideals 
$\{(C,S_1,I_1,J_1),(C,S_2,I_2,J_2)\}$ now. 
We keep repeating till we have a chain of length $m$ where 
$L(C)\subseteq sp(S_m \cup I_m)$.

Note that $S_m \cup I_m = \bigcup_{i \leq m} I_m$. Each $I_i$ is 
generated by at most $(k-2)$ forms, so there is a basis for $L(C)$ 
having at most $(k-2)m$ forms.
\end{proof}

We come to a stronger version of the main theorem of this paper.

\begin{theorem} \label{thm-chain-len}
If $C$ is a simple and minimal $\sps(k,d)$ identity then  
the length of any maximal chain of form-ideals for $C$ is at most 
$\binom{k}{2}(\log_2 d+3) + (k-1)$.
\end{theorem}

This theorem with Fact~\ref{fac-maximal-chain} imply
the main result, Theorem~\ref{thm-main}. We
prove this theorem in the next section.

\subsection{Counting all Matchings: Proof of Theorem~\ref{thm-chain-len}}\label{subsec-count}

Let a maximal chain of form-ideals $\calT$ for $C$ be $\{(C,S_1,I_1,J_1),\ldots$,
$(C,S_m,I_m,J_m)\}$. We will partition the elements of the chain
into three types according to properties of the matchings that they represent.
Each of these types will be counted separately.

We first set some notation before explaining the different types.
Let the $m$ matchings data be:
$$mdata(C,S_i,I_i,Q_i) =: \left\{(\tau_{i,q},V_i,V_{i,q})\mid q\in Q_i\right\}$$ 
We will use $mdata_i$ as shorthand for the above. For all $q\in Q_i$, $V_{i,q}$ is a sublist of 
$L(T_q)$ and $\tau_{i,q}$ is an ordered matching between $V_i, V_{i,q}$ by 
$I_i$. By the definition of useful-ness of form-ideal $I_i$ we have that
$V_{i,q}$ is disjoint to $sp(S_i\cup I_i)\setminus sp(I_i)$. Thus, 
$V_{i,q}$ can be partitioned into two sublists:
\begin{align*}
V_{i,q,0} &:= \(\ell\in V_{i,q}\mid \ell\in sp(I_i)\), \text{~~and} \\
V_{i,q,1} &:= \(\ell\in V_{i,q}\mid \ell\not\in sp(S_i\cup I_i)\).
\end{align*}
and analogously $V_i$ can be partitioned into two sublists $V_{i,0}$ and 
$V_{i,1}$. It is easy to see that these partitions induce a corresponding 
partition of $\tau_{i,q}$ as $\tau_{i,q,0}\sqcup\tau_{i,q,1}$, where 
$\tau_{i,q,0}$ (and $\tau_{i,q,1}$) is an ordered matching between 
$V_{i,0}$, $V_{i,q,0}$ (and $V_{i,1}$, $V_{i,q,1}$) by $I_i$.

Here are the three types of $mdata_i$'s:

\begin{enumerate}
	\item {\bf [Type 1]} There exist $q_1, q_2 \in Q_i$ such that
	$V_{i,q_1,1}$ is not similar to $V_{i,q_2,1}$.
	
	\item {\bf [Type 2]} There exist $q_1, q_2 \in Q_i$ such that
	$V_{i,q_1}$ is not similar to $V_{i,q_2}$, but for all 
	$r_1,r_2 \in Q_i$, $V_{i,r_1,1}$ and $V_{i,r_2,1}$ are similar.
	
	\item {\bf [Type 3]} For all $q_1, q_2 \in Q_i$, $V_{i,q_1}$
	is similar to $V_{i,q_2}$. In other words, $mdata_i$ is trivial.
\end{enumerate}

We partition $[m]$ into sets $N_1, N_2, N_3$, which are
the index sets for the $mdata$ of types $1,2,3$ respectively.

\subsubsection{Bounding $\#N_1$ and $\#N_2$}

The dominant term in Theorem~\ref{thm-chain-len} comes
from $\#N_1$. If $\#N_1$ is large, then by an averaging argument, 
for some pair $(a,b)$,
we find many matchings between forms in $T_a$ and $T_b$.
These are all orthogonal matchings, but are defined
on \emph{different} sublists of $L(T_a)$ and $L(T_b)$.
Nonetheless, we can find two dissimilar lists that
are matched too many times. Invoking Lemma~\ref{lem-general-DS}
gives us the required bound.

\begin{lemma}\label{lem-step1-1}
$\#N_1\leq\binom{k}{2}(\log_2 d+2)$.
\end{lemma}
\begin{proof}
For the sake of contradiction, let us assume $\#N_1>$ 
$\binom{k}{2}(\log_2 d+2)$. For each $mdata_i$ $(i \in N_1)$,
choose an unordered pair of indices $P_i = \{q_1,q_2\}$ such 
that $V_{i,q_1,1}$ and $V_{i,q_2,1}$ are not similar.
As there can be only $\binom{k}{2}$ distinct pairs, we get
by an averaging argument that, $s>(\log_2 d+2)$ of the $P_i$'s are equal.
Let $P_{i_1}=\cdots=P_{i_s}=\{a,b\}$ for $i_1<\cdots<i_s\in N_1$. Now 
we will focus our attention solely on the ordered matchings $\mu_i:=$ 
$\tau_{i,b,1}\tau_{i,a,1}^{-1}$ between $V_{i,a,1}, V_{i,b,1}$ by $I_i$, 
for all $i\in\{i_1,\ldots,i_s\}$. The source of contradiction is the fact 
that all these matchings are also well defined on the `last' pair of sublists 
$V_{i_s,a,1}, V_{i_s,b,1}$:

\begin{claim}\label{clm-step1}
For all $i\in\{i_1,\ldots,i_s\}$, $\mu_i$ induces an ordered matching 
between $V_{i_s,a,1}, V_{i_s,b,1}$ by $I_i$.
\end{claim}
\claimproof{clm-step1}{
%
The claim is true for $i=i_s$ so let $i<i_s$. 
The matching $\mu_i$ is an ordered $I_i$-matching between $V_{i,a,1}$, 
$V_{i,b,1}$. 
For $\ell\in V_{i_s,a,1}$, $\ell\not\in$ $sp(S_{i_s}\cup I_{i_s})$. 
Since $i<i_s$ and $L(T_a)\setminus V_{i,a,1} \subset sp(S_i \cup I_i)$,
$\ell$ cannot be in $L(T_a)\setminus V_{i,a,1}$.
Therefore, $\ell$ is in $V_{i,a,1}$.
So $\mu_i$ maps $\ell$ to some element in 
$V_{i,b,1}$, showing $\mu_i$ is defined on 
the {\em domain} $V_{i_s,a,1}$.

So we know $\mu_i$ maps $\ell\in V_{i_s,a,1}$ to an element $\mu_i(\ell)\in$ 
$V_{i,b,1}$. As $\mu_i$ is an $I_i$-matching, $\mu_i(\ell)=$
$(c\ell+\alpha)$ for some $c\in\FF^*$ and $\alpha\in sp(I_i)$ 
$\subseteq sp(I_{i_s})$, thus $\mu_i(\ell)\not\in$ 
$sp(S_{i_s}\cup I_{i_s})$ (recall $\ell\not\in$ $sp(S_{i_s}\cup I_{i_s})$). 
Thus $\mu_i(\ell)$ cannot be in $L(T_b)\setminus V_{i_s,b,1}$ (which 
has all its elements in $sp(S_{i_s}\cup I_{i_s})$). As to begin with 
$\mu_i(\ell)\in L(T_b)$ we get that $\mu_i(\ell)\in V_{i_s,b,1}$. 

Thus, $\mu_i$ maps an arbitrary $\ell\in V_{i_s,a,1}$ to 
$\mu_i(\ell)\in V_{i_s,b,1}$. In other words, $\mu_i$ induces an 
ordered matching between $V_{i_s,a,1}$, $V_{i_s,b,1}$ by $I_i$.
}

This claim means that there are $s>(\log_2 d+2)$ bipartite matchings 
between $V_{i_s,a,1}$, $V_{i_s,b,1}$ by orthogonal form-ideals 
$I_{i_1},\ldots,I_{i_s}$ respectively. Lemma \ref{lem-general-DS}
implies that the lists $V_{i_s,a,1}, V_{i_s,b,1}$ are similar. This 
contradicts the definition of $P_{i_s}$. Thus, 
$\#N_1\leq\binom{k}{2}(\log_2 d+2)$.
\end{proof}

For dealing with $\#N_2$, we use a slightly different
argument to get a better bound. We show that a Type 2
matching can involve a pair of terms at most once.

\begin{lemma}\label{lem-step1-2}
$\#N_2\leq\binom{k}{2}$.
\end{lemma}
\begin{proof}
For the sake of contradiction, assume $\#N_2>\binom{k}{2}$. 
For each $mdata_i$ ($i \in N_2$), let $P_i$ be an unordered pair
$(q_1,q_2)$ such that $V_{i,q_1}$ is not similar to $V_{i,q_2}$.
Note that because $V_{i,q_1,1}$ is \emph{similar} to $V_{i,q_2,1}$,
it must be that $V_{i,q_1,0}$ is not similar to $V_{i,q_2,0}$.
By the pigeon-hole principle, at least two $P_i$'s are the same. 
Suppose $P_{i_1}=P_{i_2}=\{a,b\}$ for 
$i_1<i_2\in N_2$.

Let $\ell\in V_{i_2,a,0}$ then by the definition of $V_{i_2,a,0}$ we have 
that $\ell\in sp(I_{i_2})$. This coupled with $i_1<i_2$ means that $\ell$ 
cannot be in $L(T_a)\setminus V_{i_1,a,1}$ (which has all its elements in 
$sp(S_{i_1}\cup I_{i_1})$). As to begin with $\ell\in L(T_a)$ we get that 
$\ell\in V_{i_1,a,1}$. Thus, $V_{i_2,a,0}$ ($V_{i_2,b,0}$) is a sublist 
of $V_{i_1,a,1}$ ($V_{i_1,b,1}$). From the useful-ness of $I_{i_2}$, the 
sublist $V_{i_2,a,0}$ ($V_{i_2,b,0}$) collects all the linear forms in 
$L(T_a)$ ($L(T_b)$) that are in $sp(I_{i_2})$ while from the useful-ness 
of $I_{i_1}$ the sublist $L(T_a)\setminus V_{i_1,a,1}$ 
($L(T_b)\setminus V_{i_1,b,1}$) is disjoint from $sp(I_{i_2})$. Thus, the
sublist $V_{i_2,a,0}$ ($V_{i_2,b,0}$) collects all the linear forms in 
$V_{i_1,a,1}$ ($V_{i_1,b,1}$) that are in $sp(I_{i_2})$. This together 
with the similarity of $V_{i_1,a,1}$ and $V_{i_1,b,1}$ gives us (by Fact 
\ref{fac-sub-similar}) that $V_{i_2,a,0}$ and $V_{i_2,b,0}$ are similar, 
which contradicts the way $P_{i_2}=\{a,b\}$ was defined. Thus, 
$\#N_2\leq\binom{k}{2}$.
\end{proof}

\subsubsection{Bounding $\#N_3$}

This requires a different argument than the pigeon-hole ideas
used for $\#N_1$ and $\#N_2$. We divide these type $3$ matchings further into 
{\em internal} and {\em external} ones. Our final aim is to prove :

\begin{lemma}\label{lem-type3}
$\#N_3 \leq (k-1)$
\end{lemma}

We shall use a combinatorial picture of how the chain of form-ideals connects 
the various multiplication terms through
matchings. We will describe an evolving \emph{forest} $\cal F$
and only deal with Type 3 $mdata_i$.

Initially, the forest $\cal F$ consists of $k$ isolated vertices, each
representing the $k$ terms $T_1,\cdots,T_k$. 
We process each $mdata_i$ in increasing order of the $i$'s, and update
the forest $\cal F$ accordingly. We will refer
to this as \emph{adding} $mdata_i$ to $\cal F$.
At any intermediate state, the forest $\cal F$ 
will be a collection of rooted trees with a total of $k$ leaves.

\begin{definition} Consider $\cal F$ when $mdata_i$
is processed. If all of $Q_i$ belongs to a single tree in $\cal F$,
then $mdata_i$ is called \emph{internal}. Otherwise, it is called \emph{external}.
\end{definition}

If $mdata_i$ is internal, $\cal F$ remains unchanged.
While each time we encounter an external $mdata_i$,
we update the forest $\cal F$ as follows. We create a new
root node labelled with $mdata_i$ (abusing notation, we refer to $mdata_i$ as a node), 
and for any tree of $\cal F$ that contains a $T_q$, $q \in Q_i$, we make
the root of this tree a child of $mdata_i$.

\begin{fact}\label{forest} 
The total number of external matchings is at most $(k-1)$.
\end{fact}
\begin{proof} 
Note that each external $mdata_i$ reduces the number of trees in the forest $\cal F$
by at least one. As initially $\cal F$ has $k$ trees and at every point of the process
it will have at least one tree, we get the claim.  
\end{proof}

It remains to count the number of internal matchings.
Whenever we encounter an internal $mdata_i$, we can
always associate it with some root $mdata_{i'}$ of $\cal F$ such that 
$i' < i$ and all of $Q_i$ is in the tree rooted at $mdata_{i'}$. 

\begin{lemma} \label{lem-internal} If $mdata_i$ is internal, then the subcircuit $C_{Q_i}$
is identically zero in $R$. Therefore, by the minimality of $C$,
no $mdata_i$ can be internal.
\end{lemma} 

This lemma with the previous fact immediately imply that $\#N_3 \leq (k-1)$.
We now set the stage to prove this lemma.
Take any Type $3$ $mdata_i$. By the triviality of $mdata_i$, the lists in  
$\{V_{i,q}\mid q\in Q_i\}$ are mutually similar. By the useful-ness of 
$I_i$ the lists in $\{L(T_q)\setminus$ $V_{i,q}\mid q\in Q_i\}$ have all 
their forms in $sp(S_i\cup I_i)\setminus sp(I_i)$.
Furthermore,
$D_i:=\sum_{q\in Q_i}$ $sc(\tau_{i,q})\alpha_q$ $M(L(T_q)\setminus V_{i,q})$ 
is a regular identity modulo $I_i$. 
Our aim is to remove the forms in $D_i$ which are common factors
(\emph{not} mod $I_i$, but mod $0$).
This gives us a new circuit (quite naturally, that will
turn out to be $sim(C_{Q_i})$) that is still an identity $(\mod I_i)$. 
In other words, start with the subcircuit $C_{Q_i}$,
and remove all common factors from this subcircuit.
This is expected to be both $sim(C_{Q_i})$ and an identity mod $I_i$.

Using this we will actually show that if $mdata_i$ is internal then 
$sim(C_{Q_i})$ is an identity $(\mod 0)$. Then we can multiply
the common factors back, and $C_{Q_i}$ would
be an absolute identity (violating minimality of $C$). 
We proceed to show this rigorously. We have
to carefully deal with field constants to ensure
that $sim(C_{Q_i})$ is indeed a factor of $D_i$.

\begin{claim} \label{clm-sim-cqi} For
Type 3 $mdata_i$, the circuit $sim(C_{Q_i})$
is an identity $\mod I_i$ and has all its forms
in $sp(S_i \cup I_i)$.
\end{claim}

\begin{proof} Let the gcd data of $D_i$ be: 
$$\vec{gcd}(D_i):=\left\{(\pi_{i,q},U_i,U_{i,q})\mid q\in Q_i\right\}$$
where $U_{i,q}$ is a sublist of $L(T_q)\setminus V_{i,q}$ and 
$\pi_{i,q}$ is an ordered matching between $U_i, U_{i,q}$ by $\{0\}$. 
Note that this is \emph{not} $\mod I_i$, even though
$D_i$ is an identity only $\mod I_i$.

By Facts \ref{fac-inv-and-union} and \ref{fac-simple} we 
can `stitch' $U$'s and $V$'s to get:
\begin{itemize}
\item $\tau_{i,q}^\prime:=$ $\tau_{i,q}\sqcup\pi_{i,q}$ is an ordered 
matching between $V_i^\prime:=V_i\cup U_i$, 
$V_{i,q}^\prime:=V_{i,q}\cup U_{i,q}$ by $I_i$. 
\item $D_i^\prime:=$ $\sum_{q\in Q_i}$ $sc(\tau_{i,q}^\prime)\alpha_q$ 
$M(L(T_q)\setminus V_{i,q}^\prime)$, is a regular identity modulo $I_i$. 
\end{itemize} 
Let $q_m$ be the minimum element in $Q_i$. We have that 
$\tau_{i,q}^\prime\tau_{i,q_m}^{\prime-1}$ is an ordered $I_i$-matching 
between the similar lists $V_{i,q_m}^\prime, V_{i,q}^\prime$. By
Fact \ref{fac-scaling1}, we can construct an ordered matching $
\mu_{i,q}$ between $V_{i,q_m}^\prime$, $V_{i,q}^\prime$ by $\{0\}$, with 
scaling factor equal to $sc(\tau_{i,q}^\prime$ 
$\tau_{i,q_m}^{\prime-1})=$ $sc(\tau_{i,q}^\prime)/sc(\tau_{i,q_m}^{\prime})$.

The way $D_i^\prime$ is constructed it is clear that $D_i^\prime$ is a 
simple circuit. This combined with the similarity of $V_{i,q_m}^\prime$, 
$V_{i,q}^\prime$ under $\mu_{i,q}$ implies that the following set of $\#Q_i$ 
matchings:
$$\left\{(\mu_{i,q},V_{i,q_m}^\prime,V_{i,q}^\prime)\mid 
q\in Q_i\right\}$$
is a gcd data of $C_{Q_i}$ modulo $(0)$ and the corresponding simple part 
is:
\begin{align*}
sim(C_{Q_i}) &= \sum_{q\in Q_i} sc(\mu_{i,q})\alpha_q
M(L(T_q)\setminus V_{i,q}^\prime)\\
&= \sum_{q\in Q_i} \frac{sc(\tau_{i,q}^\prime)}{sc(\tau_{i,q_m}^\prime)}
\alpha_q M(L(T_q)\setminus V_{i,q}^\prime)\\
&= \frac{1}{sc(\tau_{i,q_m}^\prime)}\cdot D_i^\prime
\end{align*}

Thus, $sim(C_{Q_i})$ is a regular identity mod $I_i$ as well. Also, by the useful-ness
of $I_i$, $sim(C_{Q_i})$ has all its forms in $sp(S_i \cup I_i)$. This completes the proof.
\end{proof}

We now use the structure of $\cal F$ to show relationships
between the various connected terms.

\begin{claim} \label{clm-subset}
At some stage, let $mdata_i$ be a root node of $\cal F$.
Let $X$ be a subset of the leaves of $mdata_i$. Then $L(sim(C_X))$ is
a subset of $sp(S_i \cup I_i)$.
\end{claim}

\begin{proof} Let the indices of all the external
Type 3 $mdata$ be (in order) $i_1, i_2,\cdots$.
We prove the claim by induction on the order
in which $\cal F$ is processed.
For the base case, let $i := i_1$.
Consider $\cal F$ just after $mdata_{i}$
is added. The leaves of $mdata_{i}$ are all in $Q_{i}$.
By Claim~\ref{clm-sim-cqi}, $L(sim(C_{Q_i})) \subset sp(S_i \cup I_i)$.
Any $X$ is a subset of $Q_{i}$. By Fact~\ref{fac-sim-2},
$L(sim(C_X)) \subset sp(S_i \cup I_i)$.

For the induction step, consider an external
$mdata_i$. When this is processed, a series
of trees rooted at $mdata_{j_1}, mdata_{j_2}, \cdots$
will be made children of $mdata_i$. 
Every $j_r$ is less than $i$. Let $Y_{r}$ denote the
leaves of the tree $mdata_{j_r}$. Note
that $Y_{r} \cap Q_i \neq \phi$.
By the induction hypothesis, $L(sim(C_{Y_{r}}))$
is a subset of $sp(S_{j_r} \cup I_{j_r})$ ($\subset sp(S_i \cup I_i))$.
Let $Z_1$ be $Q_i \cup Y_1$.
By Fact~\ref{fac-sim-3} applied to $sim(C_{Y_1})$ and
$sim(C_{Q_i})$, we have that $L(sim(C_{Z_1}))$ is
in $sp(S_i \cup I_i))$. Let $Z_2$ be $Z_1 \cup Y_2$.
We can apply the same argument to show that $L(sim(C_{Z_2}))$
is in $sp(S_i \cup I_i))$. With repeated applications,
we get that for $Z = \bigcup_r Y_r$, $L(sim(C_{Z})) \subset sp(S_i \cup I_i))$.
Note that $Z$ is the set of all leaves of the tree
rooted at $mdata_i$. By Fact~\ref{fac-sim-2},
$L(C_X) \subset sp(S_i \cup I_i)$,
completing the proof.
\end{proof}
We are finally armed with all the tools to prove Lemma~\ref{lem-internal}.
\begin{proof} (of Lemma~\ref{lem-internal}) Consider some
internal $mdata_i$. All the elements of $Q_i$ are leaves
in the tree rooted at some $mdata_j$, for $j < i$.
By Claim~\ref{clm-subset}, $L(sim(C_{Q_i})) \subset sp(S_j \cup I_j)$.
But by Claim~\ref{clm-sim-cqi}, $sim(C_{Q_i}) \equiv 0 \ (\mod I_i)$.
Since $I_i$ is orthogonal to $sp(S_j \cup I_j)$, Fact~\ref{fac-orth-ideals}
tells us that $sim(C_{Q_i})$ is an identity (mod $0$). Therefore, $C_{Q_i}$
is an identity.
\end{proof}

\subsection{Factors of a $\sps(k,d)$ Circuit: Proof of Theorem \ref{thm-factors}}

The ideal matching technique is quite robust and can be
used to prove Theorem~\ref{thm-factors}. Let $C$ 
be a simple, minimal, nonzero circuit with top fanin $k$ and
degree $d$ (so the different terms may have different degrees)
that computes a polynomial $p(x_1,\cdots,x_n)$. 
We remind the reader of the definition of $L(p)$.
Let us factorize
$p$ into $\prod_i q_i$, where each $q_i$ is irreducible.
Then $L(p)$ denotes the set of \emph{linear factors}
of $p$ (that is, $q_i \in L(p)$ if $q_i$ is linear).

For any $q \in L(p)$, $C \equiv 0 \ (\mod q)$, therefore
we can generate a form-ideal useful in $C$ involving $q$. Using these
we can create a chain of form-ideals whose span contains $L(p)$, and all
our counting lemmas for the matchings of types $1,2,3$ will follow. As a result,
we get a bound of $(k^3\log d)$ on the rank of $L(p)$.

\section{Concluding Remarks}

It would be very interesting to leverage the matching technique to design
identity testing algorithms. By unique factorization, matchings can be easily detected in 
polynomial time,
and it is also not hard to search for $I$-matchings involving
a specific set of forms in $I$.
We prove that depth-$3$ identities exhibit structural properties
described by the ideal matchings. Can we reverse these theorems?
In other words, can we show that certain collections of matchings are present
iff $C$ is an identity? 
This would lead to a polynomial time identity
tester for \emph{all} depth-$3$ circuits.

\comment{The algorithm of~\cite{ks07} uses a recursive procedure
based on Chinese Remaindering for identity testing.
Unfortunately, the recursion tree blows up very fast
and leads to a running time exponential in $k$. 
These matching techniques might 
allow us to design an identity tester that
can somehow overcome this barrier.}

There is still a gap between our upper bound for the rank of $O(k^3\log d)$
and the lower bound of $\Omega(k\log d)$. We feel that
$k\log d$ is the right answer and a more careful analysis of the matchings
could prove this. More interestingly, it is conjectured that when the
characteristic of the base field is $0$, the rank
is $O(k)$, \emph{independent} of $d$. We believe that
an adapation of our matching techniques to characteristic $0$ fields
could lead to such a bound.

\bibliographystyle{alpha}
\bibliography{rank-bound}

\begin{thebibliography}{Agr05}

\bibitem[AB03]{AB03}
M.~Agrawal and S.~Biswas.
\newblock Primality and identity testing via chinese remaindering.
\newblock {\em JACM}, 50(4):429--443, 2003.

\bibitem[Agr05]{Ag05}
M.~Agrawal.
\newblock Proving lower bounds via pseudo-random generators.
\newblock In {\em Proceedings of the 25th Annual Foundations of Software
  Technology and Theoretical Computer Science (FSTTCS)}, pages 92--105, 2005.

\bibitem[AV08]{AV08}
M.~Agrawal and V.~Vinay.
\newblock Arithmetic circuits: A chasm at depth four.
\newblock In {\em Proceedings of the 49th Annual Foundation of Computer Science
  (FOCS)}, 2008.

\bibitem[CK00]{CK00}
Z.~Chen and M.~Kao.
\newblock Reducing randomness via irrational numbers.
\newblock {\em SIAM J. on Computing}, 29(4):1247--1256, 2000.

\bibitem[DS06]{DS06}
Z.~Dvir and A.~Shpilka.
\newblock Locally decodable codes with 2 queries and polynomial identity
  testing for depth 3 circuits.
\newblock {\em SIAM J. on Computing}, 36(5):1404--1434, 2006.

\bibitem[KI04]{KI04}
V.~Kabanets and R.~Impagliazzo.
\newblock Derandomizing polynomial identity tests means proving circuit lower
  bounds.
\newblock {\em Computational Complexity}, 13(1):1--46, 2004.

\bibitem[KS01]{ks01}
A.~Klivans and D.~Spielman.
\newblock Randomness efficient identity testing of multivariate polynomials.
\newblock In {\em Proceedings of the 33rd Annual Symposium on the Theory of
  Computing (STOC)}, pages 216--223, 2001.

\bibitem[KS07]{ks07}
N.~Kayal and N.~Saxena.
\newblock Polynomial identity testing for depth 3 circuits.
\newblock {\em Computational Complexity}, 16(2):115--138, 2007.

\bibitem[KS08]{KSh08}
Z.~Karnin and A.~Shpilka.
\newblock Deterministic black box polynomial identity testing of depth-3
  arithmetic circuits with bounded top fan-in.
\newblock In {\em Proceedings of the 23rd Annual Conference on Computational
  Complexity (CCC)}, pages 280--291, 2008.

\bibitem[LV98]{LV98}
D.~Lewin and S.~Vadhan.
\newblock Checking polynomial identities over any field: Towards a
  derandomization?
\newblock In {\em Proceedings of the 30th Annual Symposium on the Theory of
  Computing (STOC)}, pages 428--437, 1998.

\bibitem[Sax08]{s08}
N.~Saxena.
\newblock Diagonal circuit identity testing and lower bounds.
\newblock In {\em Proceedings of the 35th Annual International Colloquium on
  Automata, Languages and Programming (ICALP)}, pages 60--71, 2008.

\bibitem[Sch80]{sch80}
J.~T. Schwartz.
\newblock Fast probabilistic algorithms for verification of polynomial
  identities.
\newblock {\em JACM}, 27(4):701--717, 1980.

\bibitem[Zip79]{Z79}
R.~Zippel.
\newblock Probabilistic algorithms for sparse polynomials.
\newblock {\em Symbolic and algebraic computation}, pages 216--226, 1979.

\end{thebibliography}


\end{document}